\documentclass[preprints,article,accept,moreauthors]{Definitions/mdpi} 
\usepackage{braket}
\usepackage{physics}
\usepackage{makecell}

\firstpage{1} 
\pubvolume{1}
\issuenum{1}
\articlenumber{0}
\pubyear{2026}
\copyrightyear{2026}
\datereceived{ } 
\daterevised{ } 
\dateaccepted{ } 
\datepublished{ } 

\Title{Security of Quantum Key Distribution with Passive Basis Choice and Detection Efficiency Mismatch for a Realistic Satellite Setup}

\Author{{Egor} 
 Ivchenko $^{1,2,*}$\orcidA{}, Aleksandr Khmelev $^{1,2}$\orcidB{}, Vladimir Kurochkin $^{1}$\orcidC{} and Anton Trushechkin $^{3}$\orcidD{}}

\AuthorNames{Egor Ivchenko, Aleksandr Khmelev, Vladimir Kurochkin and Anton Trushechkin}

\address{%
$^{1}$ \quad {NTI Center for Quantum Communications, National University of Science and Technology MISIS} 
Moscow 119049, Russia; \\
$^{2}$ \quad QSpace Technologies, Moscow 121205, Russia\\
$^{3}$ \quad {Institute for Communications Technology, Technische Universität} Braunschweig, 38106 Braunschweig, Germany;}

\corres{Correspondence: ivchenko-egor@mail.ru}

\abstract{Detection efficiency mismatch is a common problem in realistic quantum key distribution (QKD) systems. 
Existing security proofs for passive basis choice provide a nonzero secret key rate for small detection efficiency mismatches only. 
Unfortunately, in realistic setups, the detection efficiency mismatch can be significant.
Here, we present a different approach for estimating the secret key rate for the BB84 protocol with passive basis choice and a detection efficiency mismatch between four threshold detectors, where the mismatch is not assumed to be small. We also adapt the decoy-state method for this case.
The suggested approach is used to estimate the secret key rate in a QKD experiment between the Micius satellite and the Zvenigorod ground station.}

\keyword{quantum cryptography; quantum key distribution; free space communication; satellite quantum communication; security proof} 

\begin{document}

\section{Introduction}
Quantum key distribution (QKD) is a technology that enables legitimate users (Alice and Bob) to establish a common secret key for the exchange of encrypted messages.
The security of QKD protocols is based on the laws of quantum physics, which provide information-theoretic security.
This means that the eavesdropper (Eve) is unable to get any valuable information about the distributed key even if Eve has unlimited computing power and is limited only by the laws of physics.
The widely used QKD protocol proposed by Bennett and Brassard in 1984 (BB84)~\cite{Bennett2014} has various security proofs~\cite{Mayers2001, Shor2000, Renner2008, Koashi2009, Tomamichel2017}. 
However, in~practice, quantum communication systems exhibit imperfections in their physical implementation that can lead to information leakage.
This is especially relevant for realistic satellite QKD systems due to extreme operating conditions and the impossibility of changes to the satellite part of the setup after its launch~\cite{Gisin2002, Xu2020, Diamanti2016, Reutov2023, Makarov2006, Jain2016}.
Therefore, security proofs of quantum communication protocols must take into account such~vulnerabilities.

In discrete-variable QKD systems, single-photon detectors are used to register photon states with some probability, known as the detection efficiency. 
Moreover, this probability is different for different detectors.
This situation is called detection efficiency mismatch.
It is typical for QKD protocols with trusted single-photon threshold detectors~\cite{Makarov2006, Jain2016, Marcomini2025}. 
At the same time, QKD protocols that enable an eavesdropper to manipulate detectors, such as measurement-device-independent QKD~\cite{Lo2012, Tamaki2012}, are inherently resistant to such attacks but have practical challenges in implementation. 
We will examine the BB84 protocol with four detectors and passive basis choice. 
That is, the~receiver employs a beam splitter rather than a random number generator to select the measurement~basis. 

Security proofs for  QKD protocols that address detection efficiency mismatch for active basis choice, wherein an active element chooses the measurement basis characteristics, are examined in Refs.~\cite{Bochkov2019, Trushechkin2022, Tupkary2025, Grasselli2025}.
To calculate a secret key rate for the passive basis choice, a~general numerical approach can be applied~\cite{Winick2018, Coles2016, Zhang2021, Nahar2026}.
Meanwhile, analytical methods for estimating the secret key in the BB84 protocol with detection efficiency mismatch can significantly speed up key processing compared to numerical approaches. 
For real-time satellite QKD sessions with a limited communication time, a~rapid privacy amplification procedure is especially important~\cite{Liao2017, Li2025}.

A recent paper~\cite{Wang2025} provided an analytical solution for the secret key length in the case of passive basis choice detection. 
However, it is based on an assumption that the detection efficiency mismatch is small and can be treated as a ``perturbation''.
Meanwhile, we have a high detection efficiency mismatch in the QKD experiment between the Micius satellite and the Zvenigorod ground station~\cite{Khmelev2024}. 
Namely, the~efficiencies of the detectors can differ by factors greater than $2.7$, and~there are also different error rates for various transmitted states. 
These effects arise due to the non-absolute identity of the optical and detector efficiencies for different communication channels of the Zvenigorod ground receiver.
Unfortunately, formulas from Ref.~\cite{Wang2025} give a key rate of zero for these practical parameters (see {Appendix}~\ref{key_rate_old}).
As a consequence, a~theory for more accurate analysis of a large detector efficiency mismatch is needed, where it is treated as ``inherent'' and not as a perturbation. {Table~{\ref{table:comparison}} compares the main approaches for dealing with detection \linebreak efficiency mismatch.}

\begin{table}[H]
\caption{{Comparison}  of the works on detection efficiency~mismatch.}
 \label{table:comparison}

\begin{adjustwidth}{-\extralength}{0cm}
\setlength{\cellWidtha}{\fulllength/6-2\tabcolsep-0in}
\setlength{\cellWidthb}{\fulllength/6-2\tabcolsep-.20in}
\setlength{\cellWidthc}{\fulllength/6-2\tabcolsep-.20in}
\setlength{\cellWidthd}{\fulllength/6-2\tabcolsep-.20in}
\setlength{\cellWidthe}{\fulllength/6-2\tabcolsep-0in}
\setlength{\cellWidthf}{\fulllength/6-2\tabcolsep--.60in}
\scalebox{1}[1]{\begin{tabularx}{\fulllength}{>{\centering\arraybackslash}m{\cellWidtha}>{\centering\arraybackslash}m{\cellWidthb}>{\centering\arraybackslash}m{\cellWidthc}>{\centering\arraybackslash}m{\cellWidthd}>{\centering\arraybackslash}m{\cellWidthe}>{\centering\arraybackslash}m{\cellWidthf}}
\toprule
\textbf{Paper} & \textbf{Basis Choice} & \textbf{Multiphoton} & \textbf{Finite Size }& \textbf{Method} & \textbf{Notes}\\

\midrule
Bochkov~et~al.~\cite{Bochkov2019} & Active & No & No & Analytical & --\\
\midrule
Trushechkin~\cite{Trushechkin2022} & Active & Yes & No & Analytical & --\\
\midrule
Nahar~et~al.~\cite{Nahar2026} & \makecell{Active /\\ Passive} & Yes & Yes & Numerical & {Calculation was provided only for active basis choice} \\
\midrule
Wang~et~al.~\cite{Wang2025} & Passive & Yes & Yes & Analytical & {It works for a small detection efficiency mismatch}\\
\midrule
This work & Passive & No & Yes & Analytical & Collective attack only \\
\bottomrule
\end{tabularx}}
\end{adjustwidth}

\end{table}

In our previous works, we described the satellite-to-ground QKD experiment and presented a preliminary formula for the secret key rate that accounts for (not necessarily small) detection efficiency mismatch~\cite{Ivchenko2025, Khmelev2024}. 
In this paper, we provide a rigorous theoretical framework for this by developing the ideas from Refs.~\cite{Bochkov2019,Trushechkin2022}.
Here, we assume that Eve knows the exact detection efficiency values, but~she is unable to manipulate them.
Also, the~study is conducted under the single-photon Bob assumption, which suggests that Eve does not add photons to pulses. In~other words, if~Alice emits a single photon state, Bob also receives at most one photon.
Moreover, we adopt the decoy state method for our case to deal with multiphoton pulses on Alice's side. 
Using the proposed method, we analyze the satellite-to-ground quantum communication channel with a detection efficiency mismatch to estimate the final key in the QKD experiment between the Micius satellite and the Zvenigorod ground station~\cite{Khmelev2024}.

\textls[-15]{Note that satellite QKD is considered to be a promising paradigm as an important component of the future ``quantum internet'' (global quantum communications)~\cite{Takenaka2017, Liao2017, Li2025, Lu_2022, Goswami_25, Bacsardi2025, Rozenman2026, Alnaseri2026, Aljebry2026, Orsucci2025}.}

The paper is organized as follows. 
In {Section}~\ref{Preliminaries}, we describe the QKD protocol, introduce a mathematical detection model and derive the form of the density matrix for distributed photons. 
In {Section}~\ref{Problem statement}, the~optimization problem for calculating the secret key rate is defined. 
In {Section}~\ref{Calculation key generation rate}, we provide an analytical solution to the optimization problem. 
In {Section}~\ref{Decoy_state}, we adopt the decoy-state method for the case of detection efficiency mismatch with the passive basis choice. 
In {Section}~\ref{Sec6}, {we present the practical data and a channel model from a satellite-to-ground QKD experiment to validate the proposed theory. 
Moreover, we provide simulation results and a comparison of different approaches for estimating the secret key rate.}
{Section}~\ref{Conclusion} includes a summary and concluding~remarks.

\section{Preliminaries} \label{Preliminaries}
We begin with a brief description of the BB84 QKD protocol with the passive basis choice and the detection efficiency~mismatch. 

Up to Section~\ref{Decoy_state}, we assume that the sender, Alice, prepares and transmits single-photon states in a two-dimensional Hilbert space $\mathcal{H}_A = \mathbb{C}^2$. 
The states can be prepared in the Z basis ($ \ket{0}$ and $\ket{1}$) or in the X basis ($\ket{+}$ and $\ket{-}$), where $\ket{\pm} = \left(\ket{0} \pm \ket{1}\right) / \sqrt{2}$. 
The first state in each basis corresponds to the encoded bit 0 and the second one to bit 1.  
The receiver, Bob, detects pulses with one photon.
Here, we consider a scenario in which Eve is unable to add photons into pulses such that if Alice's output is a single photon, Bob's input is at most a single photon.
Then, Bob's Hilbert space is spanned by the single-photon and vacuum vectors $\ket{0}, \ket{1}$ and $\ket{\rm vac}$, i.e.,~$\mathcal{H}_B=\mathbb{C}^3$.

{If the eavesdropper (Eve) can add photons to the pulses, she gains an additional degree of freedom to control the detection probabilities at Bob’s receiver and potentially obtain extra information about the secret key. For~example, Eve can perform an intercept--resend attack: she first measures Alice’s state and then resends a pulse corresponding to her measurement result. By~increasing the photon number of the resent pulse and choosing a state that preferentially triggers the detector with the highest detection efficiency, Eve can increase the probability that Bob records the same bit as Eve. Consequently, Eve can obtain partial information about the raw key while introducing only a limited increase in the observed error rate, which may allow the attack to remain undetected. The~necessity of considering the multi-photon case has been repeatedly emphasized in Refs.\mbox{~\cite{Bochkov2019, Zhang2021, Nahar2026, Zhang2017_DEM}}. For~instance, the~probability of simultaneous photon detections at two or more detectors  (referred to as cross clicks) could be used to estimate the number of multi-photon pulses\mbox{~\cite{Zhang2021, Zhang2017_DEM}}.}

{In addition, we assume that detection efficiencies are constant and Eve cannot manipulate them. If~these assumptions do not hold, an~eavesdropper can gain additional information through various methods, including time-shift attacks, basis-dependent detection, dead-time attacks and others \mbox{\cite{Xu2020, Wang2025, Sajeed2015}}.}

We will use an equivalent entanglement-based formulation of the BB84 \linebreak protocol {\mbox{\cite{Shor2000, BBM92, Ekert91}}}.
According to this, the~initial entangled state takes the following {form:}
\begin{equation} 
\rho_{AB} = \ket{\Phi}\bra{\Phi},
\end{equation}
where $\ket{\Phi} = \frac{1}{\sqrt{2}}\left( \ket{0}_A\ket{0}_B + \ket{1}_A\ket{1}_B\right)$.

Bob's equipment contains a beamsplitter that provides a measurement in the Z basis with probability $p_z$ and in the X basis with probability $p_x = 1 - p_z$. As~a result, we add two quantum coins $ \widetilde{A}$ and $ \widetilde{B}$ to the initial state, which are responsible for the choice of the measurement bases for Alice and Bob, respectively:
\begin{equation} \label{rhoAB}
\rho_{AB\widetilde{A}\widetilde{B}} =  \rho_{AB}\otimes P\left[ \sqrt{p_z}\ket{z}+\sqrt{p_x}\ket{x}\right]^{\otimes 2},
\end{equation}
where we have introduced the notation $P\left[ \varphi\right]=\ket{\varphi}\bra{\varphi}$.

We consider QKD systems with the passive basis choice and four single-photon detectors on the receiver side with the efficiencies $\eta_{a, \alpha}, a~\in \{ z, x\}, \alpha \in \{0,1\}$. 
According to article~\cite{Zhang2017_DEM}, the~losses from the receiver optical scheme can be incorporated into the detection efficiencies of the detectors.
Then, Bob's measurement is described by the positive operator-valued measure (POVM) with the operators
\begin{equation} \label{BobPOVM}
\begin{gathered}
\begin{aligned}
P_{z,0}^B &= \eta_{z,0} \ket{0}_{B}\bra{0} \otimes \ket{z}_{\widetilde{B}}\bra{z},& \qquad 
P_{z,1}^B &= \eta_{z,1} \ket{1}_{B}\bra{1} \otimes \ket{z}_{\widetilde{B}}\bra{z}, \\
P_{x,0}^B &= \eta_{x,0} \ket{+}_{B}\bra{+} \otimes \ket{x}_{\widetilde{B}}\bra{x},& \qquad 
P_{x,1}^B &= \eta_{x,1} \ket{-}_{B}\bra{-} \otimes \ket{x}_{\widetilde{B}}\bra{x},
\end{aligned} \\
P_{\varnothing}^B = I - P_{z,0}^B - P_{z,1}^B - P_{x,0}^B - P_{x,1}^B.
\end{gathered}
\end{equation}

{In the equivalent entanglement-based formulation, Alice's preparation of quantum states is mathematically described as an ideal measurement of her subsystem of the entangled state. Thus, Alice's POVM has the same form but with perfect efficiencies:}
\begin{equation} \label{AlicePOVM}
\begin{aligned}
P_{z,0}^A &= \ket{0}_{A}\bra{0} \otimes \ket{z}_{\widetilde{A}}\bra{z}, &
P_{z,1}^A &= \ket{1}_{A}\bra{1} \otimes \ket{z}_{\widetilde{A}}\bra{z}, \\
P_{x,0}^A &= \ket{+}_{A}\bra{+} \otimes \ket{x}_{\widetilde{A}}\bra{x}, &
P_{x,1}^A &= \ket{-}_{A}\bra{-} \otimes \ket{x}_{\widetilde{A}}\bra{x}.
\end{aligned}
\end{equation}

After the measurement, Alice and Bob discard "no-detection" events and all positions whose measurement bases for Alice and Bob do not match. 
The effect of the measurement, i.e.,~the recording of its outcomes, and~sifting is described by the map~\cite{Bochkov2019, Trushechkin2022, Winick2018, Coles2016}:
\begin{equation} \label{rho2}
\begin{split}
\rho_{AB\widetilde{A}\widetilde{B}} &\rightarrow 
\sum\limits_{a \in \left\{z, x\right\}} M^A_a \otimes M^B_{a} \rho_{AB\widetilde{A} \widetilde{B}} \left(M^A_a \otimes M^B_{a} \right)^{\dagger}
=\rho_{A B \widetilde{A} \widetilde{B}}^{\rm sifted} =
\mathcal{M}(\rho_{  A B \widetilde{A}\widetilde{B}}),
\end{split}
\end{equation}
where
\begin{equation}
\begin{aligned}
M^A_z &= \sum\limits_{\alpha \in \left\{0, 1\right\}} P_{z, \alpha}^A, \qquad &
M^A_x &= \sum\limits_{\alpha \in \left\{0, 1\right\}} \mathbb{H}_A P_{x, \alpha}^A, \\
M^B_z &= \sum\limits_{\alpha \in \left\{0, 1\right\}} \sqrt{P_{z, \alpha}^B}, \qquad &
M^B_x &= \sum\limits_{\alpha \in \left\{0, 1\right\}} \mathbb{H}_B \sqrt{P_{x, \alpha}^B},
\end{aligned}
\end{equation}
and $\mathbb{H}_{a}$ is the Hadamard transformation that acts on the subspace $a$. 
We use the Hadamard operator to represent the measurement results in the set of the bit values $\{0, 1\}$ for both~bases.

Note that we can merge the registers with the measurement bases $\widetilde{A}$ and $\widetilde{B}$ into $\widetilde{\textbf{B}}$, because, after~the sifting, these registers contain the same values and are effectively two-dimensional:
\begin{equation} 
\ket{z}_{\widetilde{\textbf{B}}}=\ket{zz}_{\widetilde{A} \widetilde{B}}, \qquad
\ket{x}_{\widetilde{\textbf{B}}}=\ket{xx}_{\widetilde{A} \widetilde{B}}.
\end{equation}

Denote the normalized state as
\begin{equation} \label{rho3_1}
\widetilde\rho_{  A B \widetilde{\textbf{B}}}^{\rm sifted} 
\equiv\widetilde\rho^{\rm sifted}
= \frac{\rho_{ A B \widetilde{\textbf{B}}}^{\rm sifted}}{ \Tr \rho_{A B \widetilde{\textbf{B}}}^{\rm sifted}} = \frac{\rho_{ A B \widetilde{\textbf{B}}}^{\rm sifted}}{ p_{\rm pass}}
\end{equation}
and $p_{\rm pass} = \Tr \rho_{A B \widetilde{\textbf{B}}}^{\rm sifted}$ is the probability of passing, i.e.,~when the basis choices of Alice and Bob coincide and Bob obtains a~detection.

We analyze the asymptotic case of a large number of the transmitted states $N$, i.e.,~$N \rightarrow \infty$. 
The key generation rate is defined as the ratio of the length of the final key to the total number of pulses $N$.

\section{Secret Key~Rate} \label{Problem statement}
Here, we present the estimation of the secret key rate as an optimization~problem.

According to the Devetak--Winter theorem~\cite{Devetak2005}, the~secret key rate for the asymptotic case is
\begin{equation} \label{K}
K = p_{\rm pass} \left[H(A|E\widetilde{\textbf{B}})_{\widetilde\rho^{\rm sifted}} - f H(A|B\widetilde{\textbf{B}})_{\widetilde\rho^{\rm sifted}} \right],
\end{equation}
where $H(X|Y)$ is the conditional von Neumann entropy for the subsystem $X$ conditioned on the subsystem $Y$, $f \geq 1$ is the error correction inefficiency, 
and $E$ is the eavesdropper's subsystem of the purification $\widetilde\rho^{\rm sifted}_{AB\widetilde{A}\widetilde{B}E}$ of the state $\widetilde\rho^{\rm sifted}_{AB\widetilde{A}\widetilde{B}}$ defined in Equations~(\ref{rho2})--(\ref{rho3_1}). 

The first term in the square brackets of Equation~(\ref{K}) defines Eve's ignorance of Alice's bit sequence. 
It is based on Eve’s knowledge of her subsystem and the announced bases. 
Using the entropic uncertainty relation, we can eliminate Eve's subsystem~\cite{Berta2010, Coles2011} (note that the register $\widetilde{\textbf{B}}$ is classical):
\begin{equation} \label{ineq}
H\left[ Z_{A}|E \widetilde{\textbf{B}} \right]_{\widetilde\rho_{Z_{A} B \widetilde{\textbf{B}}}^{\rm sifted}} + H\left[ X_{A}|B \widetilde{\textbf{B}} \right]_{\widetilde\rho_{X_{A} B \widetilde{\textbf{B}}}^{\rm sifted}} \geq 1,
\end{equation}
where
\begin{equation} \label{rhoZX}
\begin{split}
\widetilde\rho_{X_{A} B \widetilde{\textbf{B}}}^{\rm sifted} &\equiv
\mathcal{X} \widetilde\rho^{\rm sifted} = \sum\limits_{\alpha \in \left\{0, 1\right\}} \mathbb{H} \ket{\alpha}_{A}\!\bra{\alpha} \mathbb{H} \widetilde\rho^{\rm sifted}\mathbb{H} \ket{\alpha}_{A}\!\bra{\alpha} \mathbb{H}, \\
\widetilde\rho_{Z_{A} B \widetilde{\textbf{B}}}^{\rm sifted} &\equiv 
\mathcal{Z} \widetilde\rho^{\rm sifted} = \sum\limits_{\alpha \in \left\{0, 1\right\}} \ket{\alpha}_{A}\!\bra{\alpha} \widetilde\rho^{\rm sifted} \ket{\alpha}_{A}\!\bra{\alpha} ,
\end{split}
\end{equation}
$\mathcal{Z}$ is the transformation of decoherence in the $z$ basis, and~$\mathcal{X}$ is the transformation of decoherence in the $x$ basis followed by the Hadamard transformation $\mathbb{H}$. 

According to Equation~(\ref{ineq}), Eve's ignorance of Alice's bit sequence can be estimated from below as follows:
\begin{equation} \label{ineq_sup}
H(A|E\widetilde{\textbf{B}})_{\widetilde\rho^{\rm sifted}} 
= H\left[ Z_{A}|E \widetilde{\textbf{B}} \right]_{\widetilde\rho_{Z_{A} B \widetilde{\textbf{B}}}^{\rm sifted}} 
\geq 1 - \sup\limits_{\rho_{AB\widetilde{A}\widetilde{B}} \in \textbf{S}} H(X_{A}| B\widetilde{\textbf{B}})_{\widetilde\rho_{X_{A} B \widetilde{\textbf{B}}}^{\rm sifted}}.
\end{equation}
{Since} Eve can control the channel and the initial state, the~supremum is taken over all initial density matrices subject to constraints on the observables:
\begin{eqnarray} \label{S}
{\textbf S} = \{ \rho_{AB\widetilde{A}\widetilde{B}} \in \mathcal{L}\left(  \mathcal{H}_A \otimes \mathcal{H}_B \otimes \mathcal{H}_{\widetilde{A}} \otimes \mathcal{H}_{\widetilde{B}} \right)  \ | \rho \geq 0, \Tr \Gamma_i \rho_{AB\widetilde{A}\widetilde{B}} = \gamma_i, i = 1, \ldots, m \},
\end{eqnarray}
where $\mathcal{L}$ is the set of linear operators on the given Hilbert space,
$\Gamma_i$ are linear operators that correspond to the observables of Alice and Bob with the observed expectation values $\gamma_i$. 

The possible observable operators $\Gamma_i$ consist of the realistic observables of Alice and Bob. We will use the following set of $\Gamma_i$:
\begin{adjustwidth}{-\extralength}{0cm}
\begin{equation} \label{G}
\begin{split}
 \Gamma_1 =  \frac{1}{p_z^2}P_{z, 0}^A  \otimes P_{z, 0}^B, \qquad
 \Gamma_2 =  \frac{1}{p_z^2}P_{z, 1}^A  \otimes P_{z, 1}^B, \qquad
 \Gamma_3 &=  \frac{1}{p_x^2}P_{x, 0}^A  \otimes P_{x, 0}^B, \qquad
 \Gamma_4 =  \frac{1}{p_x^2}P_{x, 1}^A  \otimes P_{x, 1}^B, \\
 \Gamma_5 =  \frac{1}{p_z^2}P_{z,1}^A  \otimes P_{z,0}^B, \qquad 
 \Gamma_6 = \frac{1}{p_z^2}P^A_{z,0} \otimes P^B_{z,1}, \qquad
 \Gamma_7 &=  \frac{1}{p_x^2}P^A_{x,1} \otimes P^B_{x,0}, \qquad
 \Gamma_8 = \frac{1}{p_x^2}P_{x,0}^A  \otimes P_{x,1}^B.
\end{split}
\end{equation}
\end{adjustwidth}

The values corresponding to the described constraints are introduced as follows:
\begin{adjustwidth}{-\extralength}{0cm}
\begin{equation} \label{tG}
\begin{split}
 \Tr\Gamma_1 \rho_{AB\widetilde{A}\widetilde{B}} = p_{z, 00}, \qquad
 \Tr\Gamma_2 \rho_{AB\widetilde{A}\widetilde{B}} = p_{z, 11}, \qquad
 \Tr\Gamma_3 \rho_{AB\widetilde{A}\widetilde{B}} &= p_{x, 00}, \qquad
 \Tr\Gamma_4 \rho_{AB\widetilde{A}\widetilde{B}} = p_{x, 11},
\\ 
\Tr\Gamma_5 \rho_{AB\widetilde{A}\widetilde{B}} = p_{z,10}, \qquad
 \Tr\Gamma_6 \rho_{AB\widetilde{A}\widetilde{B}} = p_{z,01}, \qquad
 \Tr\Gamma_7 \rho_{AB\widetilde{A}\widetilde{B}} &= p_{x,10}, \qquad
 \Tr\Gamma_8 \rho_{AB\widetilde{A}\widetilde{B}} = p_{x,01}.
\end{split}
\end{equation}
\end{adjustwidth}

The second term in the square brackets of Equation~(\ref{K}) defines Bob's ignorance of Alice's bit sequence. 
It is based on Bob's measurement results and open information about the key. 
{This term corresponds to the error correction in Bob's sequence. According to Fano's inequality, we can estimate this term from above as $fh(Q)$, where $h(x)=-x\log_2x-(1-x)\log_2(1-x)$ is the binary Shannon entropy, $f \geq 1$ is the error correction \linebreak inefficiency and}
\begin{equation} \label{Q}
Q = \frac{1}{p_{\rm pass}}\Tr \big[ \rho_{AB\widetilde{A}\widetilde{B}} \big( P_{z,0}^A \otimes P_{z,1}^B + P_{z,1}^A \otimes P_{z,0}^B  + P_{x,0}^A \otimes P_{x,1}^B + P_{x,1}^A \otimes P_{x,0}^B \big) \big]
\end{equation} 
is the quantum bit error rate (QBER), {because the measurement operators on Alice's side $P_{a,\alpha}^A$ and Bob's side $P_{a,1-\alpha}^B$ correspond to the opposite bits $\alpha$ and $1-\alpha$.}

Thus,
\begin{equation} \label{Knew}
K \geq p_{\rm pass} \left[ 1 - \sup\limits_{\rho_{AB\widetilde{A}\widetilde{B}} \in \textbf{S}} H(X_{A}| B\widetilde{\textbf{B}})_{\widetilde\rho_{X_{A} B \widetilde{\textbf{B}}}^{\rm sifted}} - f h(Q) \right].
\end{equation}

The conditional entropy from Equation~(\ref{Knew}) can be written as follows:
\begin{equation} \label{HD}
H(X_{A}|B\widetilde{\textbf{B}})_{\widetilde\rho_{X_{A} B \widetilde{\textbf{B}}}^{\rm sifted}} = - D\left( \mathcal{X} \widetilde\rho^{\rm sifted} || I \otimes \widetilde{\rho}_{B\widetilde{\textbf{B}}}^{\rm sifted} \right),
\end{equation}
where $\widetilde{\rho}_{B\widetilde{\textbf{B}}}^{\rm sifted} = \Tr_{A} \mathcal{X} \widetilde\rho^{\rm sifted}$ and 
$D(\sigma||\tau) = \Tr\sigma \log_2 \sigma - \Tr \sigma \log_2 \tau$ is the quantum relative entropy, which is known to be jointly convex with respect to its~arguments.

Hence, calculation of the secret key rate is reduced to solving the optimization problem in Equation~(\ref{Knew}) over the space of all possible initial density matrices $\textbf{S}$ with the constraints $\Gamma_i$, where $i\in \{1, \dots,8\}$.
\section{Estimations of the Key~Rate} \label{Calculation key generation rate} 
In this section, we provide an analytical solution to the optimization problem in Equation~(\ref{Knew}) subject to the constraints given by Equations~(\ref{S})--(\ref{tG}) {for the asymptotic and finite key  analyses}.

\begin{Theorem}
\textls[-15]{ The {asymptotic} secret key rate [Equation~(\ref{Knew})] subject to the constraints  [Equation~(\ref{G})] is lower bounded by}
\begin{equation} \label{KnewnewH}
K \geq  \sum\limits_{a \in \{ x, z \}} p_a^2 p_{{\rm pass},a}\Bigg[h\left(\frac{1-\delta_{a,a}}{2}\right) - h\left(\frac{1-\sqrt{\delta_{a,a}^2+\delta_{a,\overline{a}}^2}}{2}\right)\Bigg]- p_{{\rm pass}}f h(Q),
\end{equation}
\begin{align}
\delta_{a,a} =& \frac{p_{a,0} - p_{a,1}}{p_{{\rm pass}, a}}, \qquad &\delta_{a,\overline{a}} =& \frac{\sqrt{\eta_{a,0} \eta_{a,1}} \left(\widetilde{p}_{\overline{a},=} - \widetilde{p}_{\overline{a},\neq}\right)}{p_{{\rm pass}, a}}, \vspace*{-12pt}
\label{tau} \\
\widetilde{p}_{\overline{a},=} =&\sum\limits_{i=0}^1 \frac{p_{\overline{a}, ii}}{\eta_{\overline{a}, i}}, \qquad &\widetilde{p}_{\overline{a},\neq}=& \sum\limits_{i=0}^1\frac{p_{\overline{a}, (1-i)i}}{\eta_{\overline{a}, i}},
\label{lambda}
\end{align}
where $\overline{a}$ is the opposite basis to $a$ (i.e., $\overline{z}=x$ and  $\overline{x}=z$), and~\begin{equation} \label{def}
p_{a,\alpha}=p_{a,\alpha \alpha} +  p_{a,(1-\alpha)\alpha}, \qquad p_{{\rm pass}, a} = p_{a, 0} + p_{a, 1}, \qquad
p_{{\rm pass}}=\sum\limits_{a\in\{x,z\}} p_a^2 p_{{\rm pass}, a}.
\end{equation}
\end{Theorem}

Note that the quantities $\delta_{a,b}$ have the following physical meaning~\cite{Trushechkin2022}. 
A measurement with imperfect efficiency can be decomposed into an attenuation dependent on the detector efficiencies followed by an ideal measurement. 
Then $\delta_{a,b}$ is the normalized difference of the probabilities of Bob's outcomes in the virtual ideal measurement in the basis $b$ after the attenuation due to the imperfect measurement in the basis $a$. 
If $a=b$, this quantity corresponds to the real observable difference, according to Equation~(\ref{tau}). 

The key rate estimation procedure consists of two steps. 
Initially, we determine the appropriate form of the density matrix $\rho_{AB \widetilde{A} \widetilde{B}}$, which is defined by restrictions imposed by the observable data [see Equation~(\ref{G})] and additional symmetries. 
These symmetries can be identified using Proposition~\ref{Proposition1} in {Appendix}~\ref{denmat}. 
Then we find the matrix eigenvalues and calculate the entropies in {Appendix}~\ref{appendix_key_rate_general}.

{To account for finite key size effects, we use the leftover hash lemma from Ref.~\mbox{\cite{tan2026makingleftoversimprovedprivacy}} and results of the von Neumann entropy estimation for the asymptotic case (see Appendix~\mbox{\ref{denmat}}). Appendix~\mbox{\ref{FiniteKey}} provides a  detailed derivation of the secret key length for this case, which is given by}\vspace{-12pt}
\begin{adjustwidth}{-\extralength}{0cm}
\begin{equation} \label{key_finite_init}
\begin{gathered}
\ell_{\mathrm{ext}}^{\varepsilon_{\mathrm{sec}}} \geq \min\limits_{0\leq\delta\leq\varepsilon_{\mathrm{sec}}} \left[M\left(1-\sup\limits_{\rho_{AB\widetilde{A}\widetilde{B}} \in \textbf{S}}\overline{H}(X_{A}|B\widetilde{\textbf{B}})_{\widetilde{\rho}^{\rm sifted}} - 2 \sqrt{\frac{1-2\log_2\delta}{M}}\log_2(5)\right) - 
\log_2 \frac{1}{3(\varepsilon_{\mathrm{sec}} +\delta-\delta^2)^2}\right] - Mfh(Q),   
\end{gathered}
\end{equation}
\end{adjustwidth}
where $\varepsilon_{\mathrm{sec}} = 10^{-10}$ is the security parameter, $M$ is the sifted key length, and~$\overline{H}(X_{A}|B\widetilde{\textbf{B}})_{\widetilde{\rho}^{\rm sifted}}$ has the form:
\begin{equation}
\begin{gathered}
\overline{H}(X_{A}|B\widetilde{\textbf{B}})_{\widetilde{\rho}^{\rm sifted}}=\sup\limits_{|\hat{p}_{a,ij} - p_{a,ij}| \leq \varepsilon} H(X_{A}|B\widetilde{\textbf{B}})_{\widetilde{\rho}^{\rm sifted}}, \qquad 
    \varepsilon = \sqrt{\frac{\ln(2/\delta)}{2M}}.
    \end{gathered}
\end{equation}

Equation~(\ref{KnewnewH}) uses detailed statistics of detector clicks for various bases and bit values.
We also derive two simpler formulas, which may provide lower key rates, but~use more aggregated statistics. 
The first one is
\begin{equation} \label{K_new_case}
K \geq 
\sum\limits_{a \in \{ x, z \}} \!\!p_a^2 p_{{\rm pass},a} \left( 1 -h\!\left(\frac{1-\Delta_a}{2} \right) \right) - \sum\limits_{a \in \{ x, z \}} {p_a^2} p_{{\rm pass},a} f h(Q),
\end{equation}
where
\begin{equation} 
\Delta_a =
\frac{\sqrt{\eta_{a,0} \eta_{a,1}}\left( \widetilde{p}_{\overline{a},=} - \widetilde{p}_{\overline{a},\neq}\right)}{p_{{\rm pass},a}}.
\end{equation}
A derivation of the key rate evaluation for this case is given in {Appendix}~\ref{denmat_1}. 

Formula (\ref{K_new_case}) still uses separate statistics for each basis.
The second simplified formula uses aggregated statistics for the two bases together, and~is expressed as
\begin{equation} \label{K_new_case_2}
K \! \geq p_{\rm pass} \left( 1 - h\!\left(\frac{1-\Delta''}{2} \right) \right) -p_{\rm pass}f h(Q),
\end{equation}
\begin{equation} 
\Delta'' = \frac{1}{p_{\rm pass}}\sum\limits_{a \in\{x,z\}} p_a^2 \sqrt{\eta_{a,0} \eta_{a,1}}\left( \widetilde{p}_{\overline{a},=} - \widetilde{p}_{\overline{a},\neq}\right).
\end{equation}
A derivation of the key rate evaluation for this case is given in {Appendix}~\ref{denmat_2}. 
A numerical comparison of formulas (\ref{KnewnewH}), (\ref{K_new_case}), and~(\ref{K_new_case_2}) is presented in Section~\ref{Key rates}.

\begin{Remark} \label{remark_1}
The secret key rate is monotonic with respect to the detection efficiency mismatch parameters $\eta_{a, \alpha},\ a\in\{x,z\}, \ \alpha \in\{0,1\}$. The~monotonicity of the key rate [see Equations~(\ref{KnewnewH}) and {(\mbox{\ref{key_finite_init}})}] over the detection efficiency mismatch parameters is proved in {Appendix}~\ref{mono}. 
This fact allows us to evaluate the secret key using only bounds for detection efficiencies $\eta_{a, \alpha}, \alpha\in\{0, 1\}, a~\in \{x, z\}$, rather than the precise values. 
It is important because, in~practice, we cannot measure the exact values of the efficiencies, but~we can easily estimate them from below.
\end{Remark}

\section{Decoy-State~Method} \label{Decoy_state}
The previous reasoning applies to strictly single-photon states, but~in practice, sources of weak coherent states are often used. 
This allows for a significant increase in the initial pulse generation rate and, consequently, the~key rate. 
However, the~BB84 protocol with weak coherent pulses becomes vulnerable to a photon number split attack~\cite{Dusek2000, Lutkenhaus_2002}, and~the decoy-state method is used to preserve protocol secrecy~\cite{Wang2005, Lo2005, Ma2005, Zhang2017_decoy}. 
Here, we will consider a method with two decoy states, using the following notation: the intensity of the signal states is $\mu_{s}=\mu$, the~intensities of the decoy states are $\mu_{d}=\nu_1$, and~$\mu_{v}=\nu_2$. They must satisfy the conditions $0 \leq \nu_2 < \nu_1$ and $\nu_2 + \nu_1 < \mu$.

In the decoy-state method, the~observables are the registered click rates $^\nu\!p_{a, \alpha}$ and bit error rates $^\nu\!p_{a, (1-\alpha)\alpha}={^\nu\!q_{a, \alpha}}$ for each basis $a \in \{ z, x\}$, bit $\alpha \in \{0,1\}$, and~state type $\nu \in \{ s, d, v\}$. 
Using these values, we can estimate the parameters of the true single-photon signal states $_1^s p_{a, \alpha}$ and $_1^s p_{a, (1-\alpha)\alpha}$, $a \in \{ z, x\}$, $\alpha \in \{0,1\}$, where $_1^s p_{a, \alpha}$ is the detection probability of a single-photon signal state for the basis $a$ and bit $\alpha$, and~$_1^s p_{a, (1-\alpha)\alpha}$ is the corresponding probability of bit error in a single-photon signal state.
Also, we can express the probability of receiving the correct bit $\alpha$ for a single-photon signal state and the basis $a$ as $_1^s p_{a, \alpha\alpha} = {_1^sp_{a, \alpha}} -  {_1^sp_{a, (1-\alpha)\alpha}}$.
To calculate the key rate with the decoy-state method, it is necessary to replace $p_{a, \alpha\alpha}$ and $p_{a, (1-\alpha)\alpha}$ in Equations~(\ref{KnewnewH}), (\ref{K_new_case}), and~(\ref{K_new_case_2}) with $_1^sp_{a, \alpha\alpha}$ and $_1^s p_{a, (1-\alpha)\alpha}$, respectively. 
Note that the leading superscript denotes the type of state, and~the leading subscript denotes the number of~photons.

This section is based on the articles~\cite{Bochkov2019, Trushechkin2022}, in~which the authors examine the decoy-state method, specifically focusing on active basis choice with detection efficiency mismatch. 
Also, we use the results of Refs.~\cite{Lo2005, Ma2005, Zhang2017_decoy}, which describe a generalized decoy-state method and derive estimates for single-photon pulses and errors in them. 
The papers~\cite{Bochkov2019, Trushechkin2022} prove that the approach provided in the article~\cite{Ma2005} does not make assumptions about the detector unbalance. 
So, the~application of the decoy-state method for detection efficiency mismatch consists of separately evaluating the data by two bases and by different outcomes for Bob. 
That gives us the opportunity to generalize existing theory to the case of passive basis choice without additional~assumptions.

Note that the right-hand side of Equation~(\ref{KnewnewH}) is a function of the variables $_1^s p_{a, \alpha}$, $_1^s p_{a, (1-\alpha)\alpha}$, $a \in \{ z, x\}$, $\alpha \in \{0,1\}$, and~we want to estimate it from below. 
The rate of single-photon signal clicks $_1^s p_{a, \alpha}$ lies in the interval between the estimated lower bound ($_1^s p_{a, \alpha}^L$) and the upper bound, which we take as all registered signal clicks ($^s p_{a, \alpha}$).
Rates of erroneous single-photon clicks $_1^s p_{a, (1-\alpha)\alpha}$ lie between the lower bound, which we take as zero, 
and~the estimated upper bound ($_1^s p_{a, (1-\alpha)\alpha}^U$).
Thus, we minimize the key rate in Equation~(\ref{KnewnewH}) over the parameters of single-photon states in the presented~intervals.

As is argued in Refs.~\cite{Bochkov2019,Trushechkin2022}, well-known decoy formulas from Ref.~\cite{Bochkov2019} do not use any assumptions on the detection efficiencies. 
So, we can directly apply them (separately for each basis and bit value).
Then the detection probabilities of single-photon states are estimated as
\begin{adjustwidth}{-\extralength}{0cm}
\begin{equation} \label{q1} 
{^s_1 p_{a, \alpha}} \geq {^s_1p_{a, \alpha}^L} = \frac{\mu^2 e^{-\mu}}{\mu \nu_1 - \mu \nu_2 - \nu_1^2 + \nu_2^2}  \left( {^d p_{a, \alpha}} e^{\nu_1} - {^v\!p_{a, \alpha}} e^{\nu_2} - \frac{\nu_1^2 - \nu_2^2}{\mu^2} \left( {^s\!p_{a, \alpha}} e^{\mu} - Y_{0,a,\alpha}^L\right)\right),
\end{equation}
\end{adjustwidth}
where
\begin{equation} \label{y0}
Y_{0,a,\alpha}^L = \textrm{max} \left( \frac{{^v\!p_{a, \alpha}} \nu_1 e^{\nu_2} - {^dp_{a, \alpha}} \nu_2 e^{\nu_1}}{\nu_1 - \nu_2}, 0\right).
\end{equation}
Using Equations~(\ref{def}) and (\ref{q1}), we can obtain $_1^sp_{{\rm pass},a}$, $_1^sp_{\rm pass}$.
Also, the~error rates in the single-photon state are estimated as
\begin{equation} \label{EQ1}
{_1^s p_{a, (1-\alpha)\alpha}} \leq {_1^s p_{a, (1-\alpha)\alpha}^U} = \left( {^dq_{a, \alpha}} e^{\nu_1} - {^v\!q_{a, \alpha}} e^{\nu_2}\right) \frac{\mu e^{-\mu}}{\nu_1 - \nu_2}.
\end{equation}

Thus, the~secret key rate for the decoy-state BB84 QKD protocol with detection efficiency mismatch and passive basis choice is given as
\begin{adjustwidth}{-\extralength}{0cm}
\begin{equation} \label{Knewnewdecoy}
K \geq \min\limits_{{_1^s p}_{a, \alpha}, {_1^s q}_{a, \alpha}} \
\sum\limits_{a \in \{ x, z \}} {_1^s p}_{{\rm pass},a} {p_a^2}
\Bigg[ h \left( \frac{1-{_1^s\delta}_{a,a}}{2} \right) - 
h\left(\frac{1-\sqrt{{{_1^s\delta}_{a,a}^2}+{_1^s\delta}_{a,\overline{a}}^2}}{2}\right)\Bigg] - 
{^s p}_{{\rm pass}} fh(Q),
\end{equation}
\end{adjustwidth}
where
\begin{equation}
_1^s{\delta_{a,a}} = \frac{_1^sp_{a,0} - {_1^sp_{a,1}}}{_1^sp_{{\rm pass}, a}},  
\qquad
{}_1^s{\delta_{a,\overline{a}}}= \frac{\sqrt{\eta_{a,0}\eta_{a,1}} \left(_1^s\widetilde{p}_{\overline{a},=} - {_1^s\widetilde{p}_{\overline{a},\neq}}\right)}{_1^sp_{{\rm pass}, a}}.
\end{equation}
The simplified formulas given by Equations~(\ref{K_new_case}) and (\ref{K_new_case_2}) are adapted to the case of weak coherent pulses and decoy-state method~analogously.

{For the case of equal detection efficiencies, $\eta_{a,\alpha}=\eta$ for $a\in\{ x,z\},\alpha\in\{0,1\}$ we have $\delta_{a,a}=0$, $\delta_{a,\overline{a}}=1-2 {^s_1e_{\overline{a}}}$ where ${e_{a}}=\frac{p_{a,01}+p_{a,10}}{p_{{\rm pass}, a}}={e}$ is percentage of errors in the single-photon pulses. As~a result, the~key rate will take the well-known form\mbox{~\cite{Ma2005}}:}
\begin{adjustwidth}{-\extralength}{0cm}
\begin{equation} \label{Knewnewdecoysimple}
K \geq
\sum\limits_{a \in \{ x, z \}} {_1^s p}_{{\rm pass},a} p_a^2
\left[ 1 - h\left(^s_1e^U\right)\right] - {^s p}_{{\rm pass}} fh(Q) ={_1^s p}_{{\rm pass}}^L\left[ 1 - h\left(^s_1e^U\right)\right] - {^s p}_{{\rm pass}} fh(Q),
\end{equation}
\end{adjustwidth}

\section{Application to the Satellite QKD~Experiment} \label{Sec6}
\unskip

\subsection{Channel~Modeling} \label{Channel modeling}
To validate our theoretical results, we will use a satellite-to-ground QKD model of the quantum communication experiment between the Micius satellite and the Zvenigorod ground station~\cite{Khmelev2024, Khmelev2023_model}.
According to these studies, link efficiency, or~the overall transmission of photons at 850 nm including detection efficiency,  can be expressed over\linebreak  time as
\begin{equation} \label{eta}
\eta_{a, \alpha}^{\textrm{ch}} (t) = \frac{\varepsilon D^2}{\left(\gamma d \right)^2}  \cdot
10^{-0.4 \varkappa \textrm{csc} \theta_{\textrm{El}} \left( 1 - 0.0012 \textrm{cot}^2 \theta_{\textrm{El}}\right)} \eta_{a, \alpha},
\end{equation}
where $\theta_{\textrm{El}} = \theta_{\textrm{El}}(t)$ is the satellite elevation angle above the horizon of the ground station, 
and $d = d(t)$ is the distance to the satellite. 
Later in the text, the~time dependence of these parameters will not be explicitly mentioned, but~it will be implied. The~other parameters are defined in Table~\ref{table:station}.

Table~\ref{table:station} provides the atmospheric conditions and optical characteristics of the Zvenigorod ground station during the QKD experiment. 
The optical efficiency of the ground receiver is included in the detector efficiencies for each receiving channel. The~optical error probabilities $e_{a,\alpha}$ characterize the average probabilities (for each receiving channel) of a photon hitting an erroneous detector due to the alignment and stability of the optical system (both the transmitter and the receiver)~\cite{Khmelev2024,Ma2005}. 
They  depend slightly on the position of the satellite. Table~\ref{table:station} provides averaged~values.

\begin{table}[H]
\caption{Parameters of the Zvenigorod ground station and the atmospheric~conditions.}
 \label{table:station}
 \centering
\setlength{\cellWidtha}{\textwidth/8-2\tabcolsep-0in}
\setlength{\cellWidthb}{\textwidth/8-2\tabcolsep-0in}
\setlength{\cellWidthc}{\textwidth/8-2\tabcolsep-0in}
\setlength{\cellWidthd}{\textwidth/8-2\tabcolsep-0in}
\setlength{\cellWidthe}{\textwidth/8-2\tabcolsep-0in}
\setlength{\cellWidthf}{\textwidth/8-2\tabcolsep-0in}
\setlength{\cellWidthg}{\textwidth/8-2\tabcolsep-0in}
\setlength{\cellWidthh}{\textwidth/8-2\tabcolsep-0in}
\scalebox{1}[1]{\begin{tabularx}{\textwidth}{>{\centering\arraybackslash}m{\cellWidtha}>{\centering\arraybackslash}m{\cellWidthb}>{\centering\arraybackslash}m{\cellWidthc}>{\centering\arraybackslash}m{\cellWidthd}>{\centering\arraybackslash}m{\cellWidthe}>{\centering\arraybackslash}m{\cellWidthf}>{\centering\arraybackslash}m{\cellWidthg}>{\centering\arraybackslash}m{\cellWidthh}}
\toprule
\multicolumn{2}{c}{\bfseries Full Beam} & 
\multicolumn{2}{c}{\bfseries Telescope} & 
\multicolumn{2}{c}{\bfseries Effective } &  
\multicolumn{2}{c}{\bfseries Atmospheric } 
\\
\multicolumn{2}{c}{{\bfseries Divergence}, \boldmath{$\mu$}\textbf{rad}} & \multicolumn{2}{c}{{\bfseries Diameter}, {\textbf{mm}}} & 
\multicolumn{2}{c}{\bfseries Area Fraction } & 
\multicolumn{2}{c}{\bfseries Extinction} \\
\midrule
\multicolumn{2}{c}{$\gamma$} & 
\multicolumn{2}{c}{$D$} & 
\multicolumn{2}{c}{$\varepsilon$} & 
\multicolumn{2}{c}{$\varkappa$ } 
\\
\multicolumn{2}{c}{$10$} & 
\multicolumn{2}{c}{$600$} & 
\multicolumn{2}{c}{$0.73$} & 
\multicolumn{2}{c}{$0.22$} \\
\midrule

\multicolumn{4}{c}{\bfseries {Detector} efficiency} & \multicolumn{4}{c}{\bfseries {Optical} error probability} \\
\midrule
$\eta_{z,0}$ & $\eta_{z,1}$ & $\eta_{x,1}$ & $\eta_{x,1}$ & 
$e_{z,0}$ & $e_{z,1}$ & $e_{x,0}$ & $e_{x,1}$ \\
14.6 \% & 8.6 \% & 11.1 \% & 5.4 \% & 0.62 \% & 0.5 \% & 1.11 \% & 1.03 \% \\
\bottomrule
\end{tabularx}}
\end{table}

We simulate the count rates for each detector separately.
Using Equation~(\ref{eta}), we can obtain the probability of receiving pulses with different intensities $^\nu\!p_{\rm pass}$, $^\nu\!p_{a, \alpha}$ and the probability of receiving errors $^\nu\!q_{a, \alpha}$, where $a \in \{ z, x\}$, $\alpha \in \{0,1\}$, $\nu \in \{ s, d, v\}$. 
The observed values are represented as
\begin{equation}
{^\nu\!p_{a, \alpha}} = \frac{ p_{\mu_\nu}}{2}\left(1 - \left( 1 - \frac{Y_0}{4}\right) e^{-\eta_{a, \alpha}^{\textrm{ch}} \mu_{\nu}} \right), \qquad {^\nu\!p_{\rm pass}} = \sum\limits_{a \in \{ z, x\}}\sum\limits_{\alpha = 0}^{1} {^\nu p_{a, \alpha}},\label{nup}
\end{equation}
\begin{equation}
{^\nu\!q_{a, \alpha}} =\frac{p_{\mu_\nu}}{2}\left(e_0 \frac{Y_0}{4} + e_{a, \alpha}\left( 1 - e^{-\eta_{a, \alpha}^{\textrm{ch}} \mu_{\nu}}\right) \right),
\label{nuq}
\end{equation} 
where $e_0 = 1/2$ is the probability of error in the vacuum pulses.
Since $\nu_2 \approx 0$, the~vacuum decoy state corresponds to the dark counts, so the error probability in this state is $1/2$.
The coefficients $p_{\mu_\nu}/2$ are responsible for the part of the pulses in the total set: $p_{\mu_\nu}$ is the probability of the quantum state preparation with the intensity $\mu_\nu$, and~$1/2$ is the probability of sending \mbox{bit 0 or 1.}
In addition, we assume that the probability of background noise $Y_0 = 3\cdot10^{-6}$ is distributed uniformly among the four detectors, which gives the probability $Y_0/4$ of a dark count in each detector.
The parameters of the QKD protocol of the satellite transmitter are presented in Table~\ref{table:param}.

\begin{table}[H]
\caption{QKD protocol parameters of the Micius~satellite.}
 \label{table:param}
 
 \centering
\begin{adjustwidth}{-\extralength}{0cm}
\setlength{\cellWidtha}{\fulllength/9-2\tabcolsep-0in}
\setlength{\cellWidthb}{\fulllength/9-2\tabcolsep-0in}
\setlength{\cellWidthc}{\fulllength/9-2\tabcolsep-0in}
\setlength{\cellWidthd}{\fulllength/9-2\tabcolsep-0in}
\setlength{\cellWidthe}{\fulllength/9-2\tabcolsep-0in}
\setlength{\cellWidthf}{\fulllength/9-2\tabcolsep-0in}
\setlength{\cellWidthg}{\fulllength/9-2\tabcolsep-0in}
\setlength{\cellWidthh}{\fulllength/9-2\tabcolsep-0in}
\setlength{\cellWidthi}{\fulllength/9-2\tabcolsep-0in}
\scalebox{1}[1]{\begin{tabularx}{\fulllength}{>{\centering\arraybackslash}m{\cellWidtha}>{\centering\arraybackslash}m{\cellWidthb}>{\centering\arraybackslash}m{\cellWidthc}>{\centering\arraybackslash}m{\cellWidthd}>{\centering\arraybackslash}m{\cellWidthe}>{\centering\arraybackslash}m{\cellWidthf}>{\centering\arraybackslash}m{\cellWidthg}>{\centering\arraybackslash}m{\cellWidthh}>{\centering\arraybackslash}m{\cellWidthi}}
\toprule
{\bfseries Repetition} & \multicolumn{2}{c}{\bfseries Bases Probability} & \multicolumn{3}{c}{\bfseries Emission Probability} &  \multicolumn{3}{c}{\bfseries Average Photon Number} \\
{\bfseries Rate}\textbf{,} \textbf{Hz} & \textbf{\boldmath{$X$ Basis}}&\textbf{\boldmath{$Y$ Basis}}&\textbf{Signal}&\textbf{Decoy}&\textbf{Vacuum}&\textbf{Signal}&\textbf{Decoy}&\textbf{Vacuum} \\
\midrule
$f$ & $p_x$ & $p_y$ & $p_{\mu_s}$ & $p_{\mu_d}$ & $p_{\mu_v}$ & $\mu$ & $\nu_1$ & $\nu_2$\\
$10^8$ & $0.5$ & $0.5$ & $0.5$ & $0.25$ & $0.25$ & $0.8$ & $0.1$ & $10^{-6}$ \\
\bottomrule
\end{tabularx}}
\end{adjustwidth}

\end{table}

Figure~\ref{fig:mod_exp} shows the results of the satellite-to-ground QKD experiment between the Micius satellite and the Zvenigorod ground station~\cite{Khmelev2024}, as~well as the count rate approximation using the described channel model. 
Note that these graphs depict the overall number of clicks in a particular detector, i.e.,~we take into account the quantum states with all intensity types. 
The count rate in this case can be expressed as follows:
\begin{equation} \label{Countrate}
{R_{a, \alpha}^{\rm{sum}}} = f\sum\limits_{\nu \in \{ s, d, v\}} {^\nu\!p_{a, \alpha}},
\end{equation}
where ${}^\nu\!p_{a, \alpha}$ are calculated according to Equations~(\ref{eta}) and (\ref{nup}). 
The detector efficiencies given in Table~\ref{table:station} are chosen so that the approximation curves defined by Equations~(\ref{eta}) and (\ref{nup}) give the best fit to the experimental data on {Figure}~\ref{fig:mod_exp}. 
From Table~\ref{table:station}, we see a high detection efficiency mismatch in our QKD experiment between the Micius satellite and the Zvenigorod ground station.
Namely, the~efficiencies for the channels $Z,0$ and $X,1$ differ by a factor of~2.7.

{The detection efficiency mismatch occurs due to imperfections in the optical setup and detectors.
Here, we assume that the detection efficiency mismatch remains constant during the satellite passage. According to Ref.~\mbox{\cite{Khmelev2024}}, the~maximal deviation of detector efficiencies from the obtained values is \mbox{$\Delta_\eta = 0.05$}, \mbox{${\eta}_{a,\alpha}\in\{(1-\Delta_\eta)\hat{\eta}_{a,\alpha}, (1+\Delta_\eta)\hat{\eta}_{a,\alpha}\}$}. Therefore, we use lower bounds on the detection efficiencies to calculate the key rate according to Appendix~\mbox{\ref{mono}}.}

\begin{Remark} \label{remark_eta}
{The detection efficiencies $\eta_{a,\alpha}$ are not measured directly. Instead, we fit Equation~(\mbox{\ref{eta}}) to the experimental data (see Figure~\mbox{\ref{fig:mod_exp}}) and obtain the detection efficiencies, which are assumed to remain constant during the satellite passage.
The overall losses depend on the satellite position, i.e.,~on time, but~within this model the detection efficiencies themselves are treated as being time-independent.
Despite the close correspondence between the model and the experiment, Eve might still deliberately manipulate the detector efficiencies to gain additional information. 
Therefore, some time-dependence of the efficiency cannot be ruled out and may be the focus of future research, but~now the assumption of constant efficiency is a reasonable one.}
\end{Remark}\vspace{-12pt}

\begin{figure}[H]
\begin{center}
\includegraphics[width=9cm]{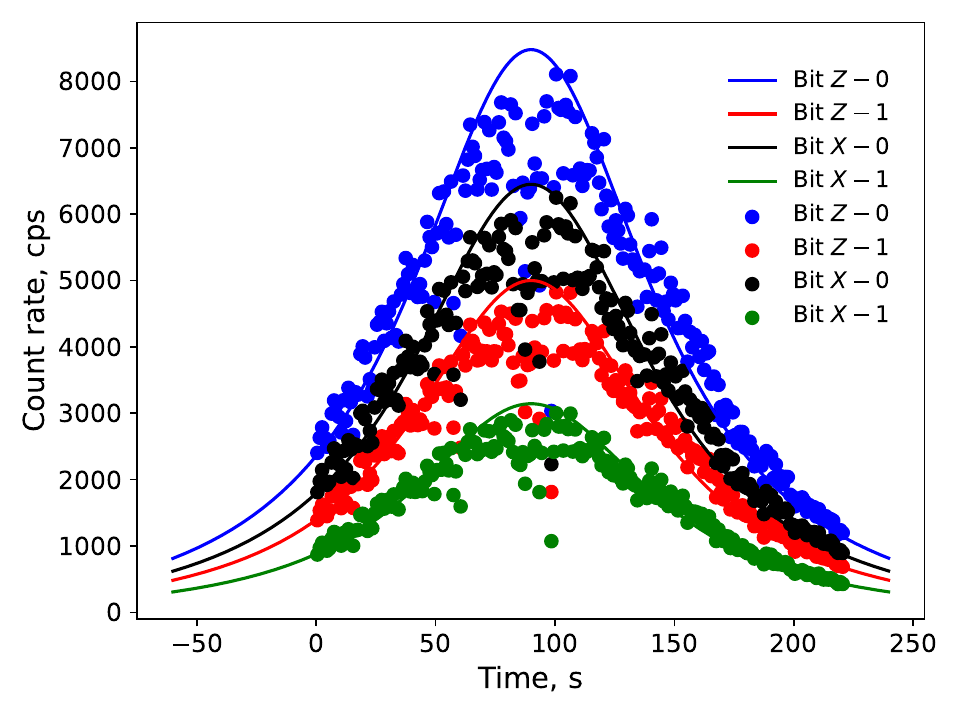}
\caption{The photon count rates for four receiving channels during the satellite-to-ground QKD experiment~\cite{Khmelev2024}.
The experimental data are shown as dots. The~solid lines represent the approximation $R_{a, \alpha}^{\rm{sum}}$ of the data for each receiving channel of the ground receiver given by Equations~(\ref{eta}), (\ref{nup}), and~(\ref{Countrate}). 
The detector efficiencies in Equation~(\ref{eta}) are chosen for best fit and are given in Table~\ref{table:station}. 
}
\label{fig:mod_exp}
\end{center}
\end{figure}

\begin{Remark} \label{remark_eta_eva}
{Our work is specifically intended for application to a particular experimental setup, and~we do not consider the scenario in which an eavesdropper can manipulate detection efficiencies. This issue can be considered for the future development of the presented theory. However, it would require additional experimental data that we do not currently possess.}
\end{Remark}

\subsection{Key~Rates}
\label{Key rates}
In this section, we calculate the secret key for the satellite-to-ground QKD experiment between the Micius satellite and the Zvenigorod ground station, taking into account the detection efficiency mismatch between four threshold detectors. 
Furthermore, we compare different formulas for the secret key~rate.

The total number of received clicks and the sifted key length in the satellite-to-ground QKD experiment~\cite{Khmelev2024} over 220 seconds are 4956 kbits and 2491 kbits, respectively.
To calculate the secret key, we use two approaches: the presented analytical Equation~(\mbox{\ref{key_finite_init}}) and a numerical method from Ref.~\mbox{\cite{Nahar2026}} with error correction inefficiency $f_{\rm ec}=1.44$ and security parameter $\varepsilon_{\textrm{EC}}=10^{-10}$.
Also, we assume that real detector efficiencies have mean values from Table\mbox{~\ref{table:station}} and lie within the interval with deviation {$\Delta_\eta=0.05$}. 
As a result, our analytical formula (\mbox{\ref{key_finite_init}}) gives the secret key length of 271,670 bits in less than one second of computational time, while the numerical approach takes more than 100 seconds and provides 277,100 bits. 
More detailed quantities, namely 
$^\nu\!p_{\rm pass}$, $^\nu\!p_{a, \alpha}$, and~$^\nu\!q_{a, \alpha}$, where $a \in \{ z, x\}$, $\alpha \in \{0,1\}$, and~$\nu \in \{ s, d, v\}$, are simulated using Equations~(\ref{eta})--(\ref{nuq}).

Figure~\ref{fig:Key_rate}a
shows the plots of the secret key rate vs time calculated using the modification of Equations~(\ref{KnewnewH}), (\ref{K_new_case}), and~(\ref{K_new_case_2}) for the decoy-state method (see Section~\ref{Decoy_state}) with an error correction inefficiency of $f_{\rm ec}=1.44$. 
We compare the secret key rate with the case of no detection efficiency mismatch, but~the same average probabilities of detection and error. 
Namely, we replace all $^\nu p_{a, \alpha}$ and $^\nu q_{a, \alpha}$ by
\begin{eqnarray} \label{etaeq}
{^\nu\!p_{eq}} =  \frac{1}{4} \sum\limits_{a \in \{ z, x\}}\sum\limits_{\alpha = 0}^{1} {^\nu\!p_{a, \alpha}} , \qquad
{^\nu\!q_{eq}} = \frac{1}{4} \sum\limits_{a \in \{ z, x\}}\sum\limits_{\alpha = 0}^{1} {^\nu\!q_{a, \alpha}}.
\end{eqnarray}
The difference in secret key rates between the no-mismatch and mismatch cases is less than 9\%. 
Also, Equation~(\ref{K_new_case}) results in a decrease in the secret key rate of less than 10\% in comparison with Equation~(\ref{KnewnewH}).

\begin{figure}[H]  
\includegraphics[width=.98\linewidth]{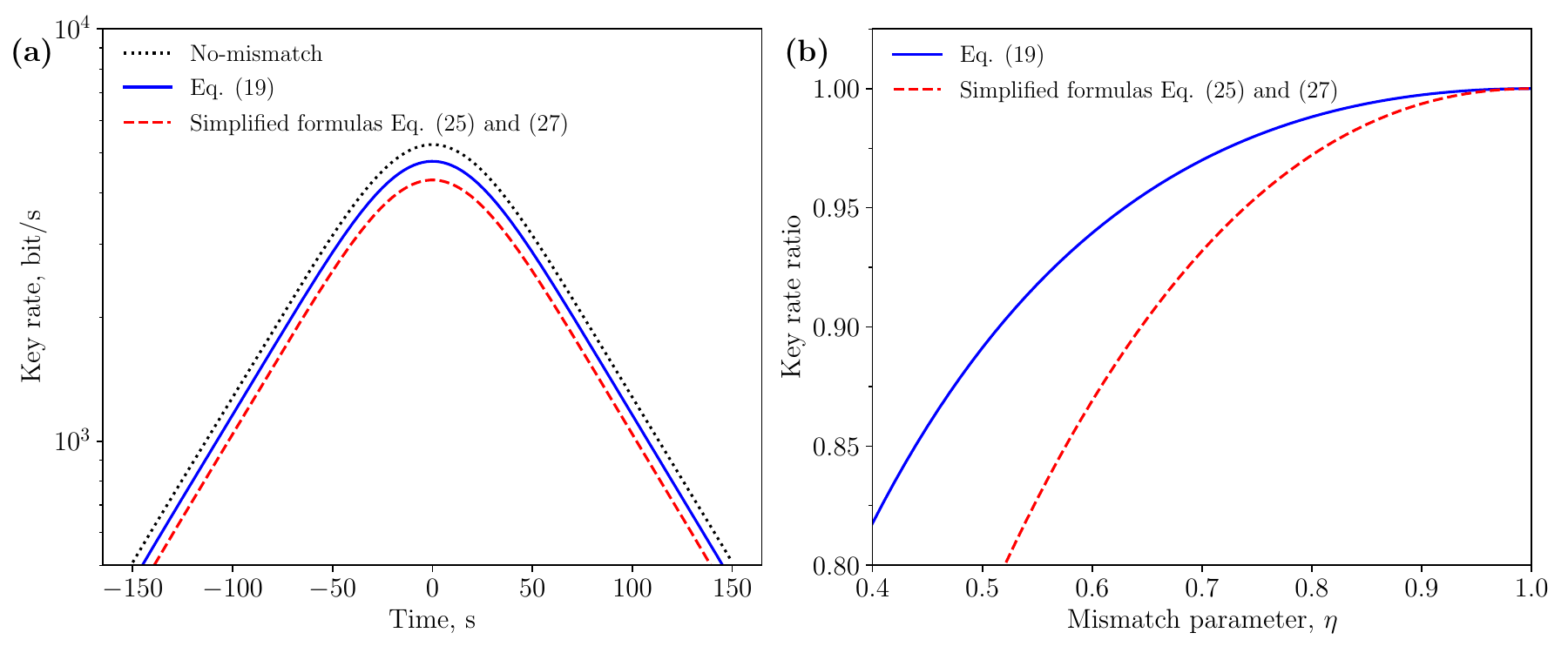}
\caption{The key rate calculation for different approaches and different detection efficiency mismatch. For~both figures, Equation~(\ref{K_new_case_2}) yields similar results to Equation~(\ref{K_new_case}) with the difference being less than 0.1\% of our data; hence, these approaches relate to the same line. (\textbf{a}) The dependence of the secret key rate over a satellite passage.
The no-mismatch case is shown by a dotted line. 
The blue solid line and red dashed line represent the detection efficiency mismatch for Equations~(\ref{KnewnewH}) and (\ref{K_new_case}), respectively.
The satellite passage is symmetrically aligned with respect to the highest elevation angle that corresponds to zero time.
(\textbf{b}) The ratio of the key generation rate with and without detection efficiency mismatch dependent on the mismatch parameter $\eta=\frac{\eta_{z,1}}{\eta_{z,0}}=\frac{\eta_{x,1}}{\eta_{x,0}}$ (note that the simplifying assumption is used only in this plot for illustrative purposes). 
The blue one corresponds to the calculation of the mismatch case by exact Equation~(\ref{KnewnewH}), and~the red one by simplified Equation~(\ref{K_new_case}).}
\label{fig:Key_rate}
\end{figure}

Figure~\ref{fig:Key_rate}b
shows the dependence of key generation rate on the detection efficiency mismatch as compared to the no-mismatch case under the simplifying assumption\linebreak  $\eta=\frac{\eta_{z,1}}{\eta_{z,0}}=\frac{\eta_{x,1}}{\eta_{x,0}}$.
{We see that exact \mbox{Equation~(\ref{KnewnewH})} and the approach based on the simplified \mbox{Equations~(\ref{K_new_case})} and \mbox{(\ref{K_new_case_2})} produce a similar key rate for $\eta>0.95$ but differ significantly for $\eta<0.95$. }

Finally, we adapt the approach from Ref.~\cite{Wang2025}, which uses the decoy-state BB84 protocol with passive basis choice and detection efficiency mismatch, to~our case of asymptotic key without multiphoton clicks.
As a result, the~formula for the secret key rate for memoryless detectors ({Equation~(B7),} 
 \cite{Wang2025}) is reduced to the following form:
\begin{gather} 
K =
\max\left(\mathcal{B}_1 \left( 1 - h\left( \mathcal{B}_e\right)\right)-\lambda_{EC}, 0\right),\\
\mathcal{B}_1 = \mathcal{B}_1^{\rm{decoy}} - q_Z,
\end{gather}
where $ \lambda_{EC}$ is the error correction term,
$\mathcal{B}_1$ is the probability of receiving a single-photon signal state, 
$\mathcal{B}_1^{\rm{decoy}}$ is the probability of receiving a single-photon signal state estimated from the decoy-state method, 
$\mathcal{B}_e$ is the error probability in the received signal single-photon states, 
$q_Z$ quantifies the fraction of 0-photon key~rounds.

Unfortunately, the~approach of Ref.~\cite{Wang2025} works only for a small detection efficiency mismatch, while, in~our practical satellite setup, $\eta_{z,0}$ is almost 2.7 times larger than $\eta_{x,1}$.
This leads to loose bounds on the phase errors and zero key rate.
The details are presented in {Appendix}~\ref{key_rate_old}.
\section{Conclusions} \label{Conclusion}
We have addressed one of the main issues for practical QKD systems---detection efficiency mismatch---and have provided an analytical approach for the secret key rate estimation in the BB84 protocol with the passive basis choice for a practical satellite QKD system with high detection efficiency mismatch. 
The main result is Equation~(\ref{KnewnewH}) for the secret key rate as well as its simplifications given in Equations~(\ref{K_new_case}) and (\ref{K_new_case_2}) providing slightly reduced rates.
{Also, we have presented Equation~(\mbox{\ref{key_finite_init}}) to calculate the secret key rate for a finite number of pulses and collective attack.}
Additionally, we have examined the decoy-state method for QKD systems that utilize passive basis choice and have detection efficiency mismatches (see Equation~(\ref{Knewnewdecoy})). 
The presented results can deal with high detection efficiency mismatch and in this sense complement a different analysis from Ref.~\cite{Wang2025}, which considers a more general scenario, but~requires small detection efficiency~mismatch.

{We have shown that the decrease in key rate due to the detection efficiency mismatch is less then 1\% whenever the ratios of the detector efficiencies $\eta=\frac{\eta_{z,1}}{\eta_{z,0}}=\frac{\eta_{x,1}}{\eta_{x,0}}$ fall within the range $[0.9, 1]$, and~in the range of experimental values $\eta \in [0.5, 0.9]$ the decrease of the key rate may be significant. Also, for~the small detection mismatch $\eta=\frac{\eta_{z,1}}{\eta_{z,0}}=\frac{\eta_{x,1}}{\eta_{x,0}}>0.95$, both the exact Equation~(\mbox{\ref{KnewnewH}}) and the simplified Equations~(\mbox{\ref{K_new_case}}) and (\mbox{\ref{K_new_case_2}}) provide the same results, so there is no need to use a more complex formula. However, for~larger detection mismatches, the~exact Equation~(\mbox{\ref{KnewnewH}}) is utilized to obtain a more accurate estimation of the secret key rate, and~the simplified Equations~(\mbox{\ref{K_new_case}}) and (\mbox{\ref{K_new_case_2}}) are used if we do not know the complete statistics for each state.}
To validate our theory, we have calculated the key rate for the satellite-to-ground QKD experiment {using the precise formula for key Equation~(\mbox{\ref{key_finite_init}})}, which is characterized by a high detection efficiency~mismatch.

The limitations of the presented paper are the asymptotic case {for Equations~(\mbox{\ref{KnewnewH}}), (\mbox{\ref{K_new_case}}) and (\mbox{\ref{K_new_case_2}}),  collective attack only for Equation~(\mbox{\ref{key_finite_init}})}, the~constant detection efficiencies, and~the single-photon approximation on Bob's side.  
The latter assumption means that the eavesdropper sends no more than one photon to the legitimate receiver if the legitimate sender’s pulse is a single photon. 
Multiple photons on Bob's side lead to detector double clicks; hence, this limitation can be overcome by counting these clicks and estimating of the number of such positions~\cite{Trushechkin2022,Wang2025}. 

Let us also emphasise that, in~our experiment, the~values of detector efficiencies are the results of fitting the experimental data rather than directly observed quantities (see Section~\ref{Channel modeling}), which stresses the importance of security analysis for variable or not fully characterised detector efficiencies~\cite{Wang2025}.

Finally, further research can be directed to the mismatch of the dark count rates. Differences between the dark count rates of the detectors may also affect the security and secret key rate~\cite{Wang2025}.
Another challenge is the variable detection efficiency mismatch during quantum key distribution. 
This is a scenario in which the detection efficiency mismatch is not constant but is partially within Eve's control. 
In other words, we will have not only the imbalance that arose due to setup imperfection but also an imbalance introduced by Eve . 
These attacks are detailed and experimentally tested in these works~\cite{Zhao2008, Sajeed2015, Pirandola2020}, and~Ref.~\cite{Grasselli2025} provides a way to estimate such imperfections for the active basis~choice.

\vspace{6pt} 





\authorcontributions{Conceptualization, E.I.; methodology, E.I. and A.T.; validation, E.I. and A.K.; investigation, E.I. and A.T.; writing---original draft preparation, {E.I.}; writing---review and editing, A.K., {V.K.}
 and A.T.; visualization, A.K. and E.I.; supervision, A.T. and V.K. All authors have read and agreed to the published version of the manuscript.}

\funding{This research was funded by Russian Science Foundation grant number 25-22-00342 (Efficiency and security of realistic satellite-to-ground quantum key distribution systems).}

\dataavailability{Data underlying the results presented in this paper are not publicly available at this time but may be obtained from the authors upon reasonable request.} 

\acknowledgments{{We thank Cameron Foreman for the fruitful discussion and the comments that helped to improve the article.}
{E.I., A.K., and V.K.} acknowledge support from the Russian Science Foundation grant No.~25-22-00342.}

\conflictsofinterest{{The} authors declare no conflicts of~interest.}

\appendixtitles{yes} 
\appendixstart
\appendix
\section{Key Generation Rate~Calculations} \label{denmat}
Here, we derive Equations~(\ref{KnewnewH}), (\ref{K_new_case}), and~(\ref{K_new_case_2}) for the key~rate.

Firstly, we determine the form of the initial density matrix $\rho_{AB}$ that satisfies the optimization problem Equation~(\ref{Knew}).
This form is determined by the constraints on the observable values $\Gamma_i$ (see Equations~(\ref{G}) and (\ref{tG})) and additional symmetries. 
These symmetries are identified in the proposition proved~below.

\begin{Proposition}
\label{Proposition1} 
Let $\Phi$ be a positive trace-preserving linear mapping on the operator space $\mathcal {L} \left( \mathcal{H}_A \otimes \mathcal{H}_B \otimes \mathcal{H}_{\widetilde{A}} \otimes \mathcal{H}_{\widetilde{B}}\right)$ satisfying the conditions $\Phi^{\dagger}\left( \Gamma_i \right) = \Gamma_i$ for all $i$. 
In addition, let $p_{\rm pass}$ be uniquely determined by the observed values $\Tr \Gamma_i \rho_{AB\widetilde{A}\widetilde{B}}$ and
\begin{equation} \label{ys}
\Phi  \left(\mathcal{X} \left( \mathcal{M}\left( \rho_{AB\widetilde{A}\widetilde{B}} \right) \right) \right) =  \mathcal{X} \left( \mathcal{M}\left(\Phi \left( \rho_{AB\widetilde{A}\widetilde{B}} \right) \right) \right),
\end{equation}
where $\mathcal{M}$ and $\mathcal{X}$ maps are defined in Equations~(\ref{rho2}) and (\ref{rhoZX}), respectively.

Then,
\begin{equation} \label{T1}
\sup\limits_{\rho_{AB\widetilde{A}\widetilde{B}} \in \textbf{S}} H(X_A| B  \widetilde{\textbf{B}})_{\widetilde\rho^{\rm sifted}} = \sup\limits_{\rho_{AB\widetilde{A}\widetilde{B}} \in \textbf{S}'} H(X_A| B \widetilde{\textbf{B}})_{\widetilde\rho^{\rm sifted}},
\end{equation}
where $\widetilde\rho^{\rm sifted}=\mathcal{M}\left( \rho_{AB\widetilde{A}\widetilde{B}}\right)$ [see Equation~(\ref{rho2})], $\widetilde{\textbf{B}}=\widetilde{A}\widetilde{B}$,
\mbox{$\textbf{S}' = \textbf{S} \bigcap  \operatorname{Im}\left( \Phi\right)$},
\textbf{S} is defined in\linebreak  Equation~(\ref{S}), and~${\operatorname{Im}}(\Phi)$ denotes the $\Phi$ image.
\end{Proposition}

The proposition is similar to that provided in Ref.~\cite{Trushechkin2022}, but~with some modifications. So, let us give a brief~proof.

\begin{proof}
From Equation~(\ref{rhoZX}), we have \hbox{$\rho^{\rm{sifted}}_{X_{A} B\widetilde{\textbf{B}}}=\mathcal{X} \mathcal{M} (\rho_{AB\widetilde{A}\widetilde{B}})$} and $\rho_{B\widetilde{\textbf{B}}}=\Tr_{A}\mathcal{X} \mathcal{M} (\rho_{AB\widetilde{A}\widetilde{B}})$.
Then, the~conditional entropy is estimated from above as \vspace{-12pt}
\begin{adjustwidth}{-\extralength}{0cm}
\begin{equation} \label{entropy_ineq}
\begin{split}
-H(X_A| B \widetilde{\mathbf{B}})_{\rho_{AB\widetilde{A}\widetilde{B}}} 
&= D\left( \rho_{X_{A} B\widetilde{\mathbf{B}}} \;\middle\|\; I \otimes \rho_{B\widetilde{\mathbf{B}}} \right) 
= D\left( \mathcal{X} \mathcal{M} (\rho_{AB\widetilde{A}\widetilde{B}}) \;\middle\|\; I \otimes \operatorname{Tr}_{A}\mathcal{X} \mathcal{M} (\rho_{AB\widetilde{A}\widetilde{B}}) \right) \\
& \geq D\left( \Phi \mathcal{X} \mathcal{M} (\rho_{AB\widetilde{A}\widetilde{B}}) \;\middle\|\; \Phi \left(I \otimes \operatorname{Tr}_{A}\mathcal{X} \mathcal{M} (\rho_{AB\widetilde{A}\widetilde{B}}) \right) \right) \\
&= D\left( \mathcal{X} \mathcal{M} (\Phi \rho_{AB\widetilde{A}\widetilde{B}}) \;\middle\|\; I \otimes \operatorname{Tr}_{A}\mathcal{X} \mathcal{M} (\Phi \rho_{AB\widetilde{A}\widetilde{B}}) \right) 
= -H(X_A| B \widetilde{\mathbf{B}})_{\Phi \left(\rho_{AB\widetilde{A}\widetilde{B}} \right)}.
\end{split}
\end{equation}
\end{adjustwidth}
Here, we use the monotonicity of the quantum relative entropy under the action of a positive trace-preserving linear map on both arguments~\cite{Muller_Hermes2017}. 
The penultimate equality follows from the condition~(\ref{ys}).

In view of inequality~(\ref{entropy_ineq}) and conditions $\Phi^{\dagger}\left( \Gamma_i \right) = \Gamma_i$ for all $i$, it follows that
\begin{equation}
\begin{gathered}
\sup\limits_{\rho_{AB\widetilde{A}\widetilde{B}} \in \textbf{S}} H(X| B \widetilde{\textbf{B}})_{\rho} \leq \sup\limits_{\rho_{AB\widetilde{A}\widetilde{B}} \in \textbf{S}'} H(X| B \widetilde{\textbf{B}})_{\rho}.
\end{gathered}
\end{equation}
The reverse inequality also holds since $\textbf{S}'$ is a subset of $\textbf{S}$. 
Taking into account that
\begin{equation}
\begin{gathered}
 H(X| B \widetilde{\textbf{B}})_{\rho} =p_{\rm pass} H(X| B \widetilde{\textbf{B}})_{\widetilde\rho},
\end{gathered}
\end{equation}
we obtain Equation~(\ref{T1}). 
The proposition has been proved.
\end{proof}

Note that if $\Phi$ is a projector, i.e.,~$\Phi = \Phi^2$ then $\textbf{S}'$ can be rewritten as
\begin{equation} \label{prop_remark}
{\bf S}' = {\bf S} \bigcap \{\rho_{AB\widetilde{A}\widetilde{B}}||\Phi\left(\rho_{AB\widetilde{A}\widetilde{B}}\right) = \rho_{AB\widetilde{A}\widetilde{B}}\}.
\end{equation}

That is, it is sufficient to optimize over the density matrices satisfying the symmetry given by Equation~(\ref{prop_remark}). 
Consider the projectors
\begin{equation}
\Phi_1 (\rho) = \frac{1}{2}\left( \rho + \rho^*\right), \qquad
\Phi_2(\rho)= \frac{1}{2}\left(\rho + R_{\Phi,2} \rho R_{\Phi,2}^+\right),
\end{equation}
where $^*$ denotes the complex conjugate of the elements of $\rho$ in the basis $z$ and
\begin{equation}
 R_{\Phi,2} = \left( \ket{z}_{\widetilde{A}}\! \bra{z} \otimes Z_A + \ket{x}_{\widetilde{A}}\! \bra{x} \otimes X_A \right)  
 \otimes \left(\ket{z}_{\widetilde{B}}\!\bra{z} \otimes Z_B + \ket{x}_{\widetilde{B}}\!\bra{x} \otimes X_B \right).
\end{equation}
These projectors satisfy Proposition~\ref{Proposition1}. 
According to Equation~(\ref{prop_remark}), the~set of states $\rho_{AB\widetilde{A}\widetilde{B}}$ is restricted to those that satisfy $\Phi_i(\rho_{AB\widetilde{A}\widetilde{B}}) = \rho_{AB\widetilde{A}\widetilde{B}}$, $i=1,2$.
Also, we only consider subsystems with $\ket{zz}_{\widetilde{A}\widetilde{B}}\!\bra{zz}$ and $\ket{xx}_{\widetilde{A}\widetilde{B}}\!\bra{xx}$ because the remaining parts are eliminated during the sifting. 
As a result, all elements of the matrix $\rho_{AB\widetilde{A}\widetilde{B}}$ are real and satisfy\linebreak  the equalities
\begin{equation}
\begin{gathered}
\left(Z\otimes Z \rho_{AB} Z\otimes Z \right) \otimes \ket{zz}_{\widetilde{A}\widetilde{B}}\!\bra{zz} =\rho_{AB} \otimes \ket{zz}_{\widetilde{A}\widetilde{B}}\!\bra{zz}, \\
\left(X\otimes X \rho_{AB} X\otimes X \right) \otimes \ket{xx}_{\widetilde{A}\widetilde{B}}\!\bra{xx} =\rho_{AB} \otimes \ket{xx}_{\widetilde{A}\widetilde{B}}\!\bra{xx},
\end{gathered}
\end{equation}
according to $\Phi_1$ and $\Phi_2$.

Using the symmetries arising from $\Phi_1$ and $\Phi_2$, and~the restrictions on observables given by Equations~(\ref{G}) and (\ref{tG}),
we obtain the density matrix after the sifting procedure (see Equation~(\ref{rho3_1})) in the following form:
\begin{equation} \label{appendix_p_sift}
p_{\rm pass} \widetilde{\rho}
^{\rm sifted}=p_{\rm pass} \mathcal{M}(\rho_{AB\widetilde{A}\widetilde{B}}) 
= \sum\limits_{a \in \{x,z\}}p_a^2\ket{a}_{\widetilde{\textbf{B}}}\!\bra{a} \otimes \rho_{AB}^{a},
\end{equation}
where
\begin{equation}
\begin{gathered}
\rho_{AB}^{a}= \begin{pmatrix}
\eta_{a,0}\widetilde{p}_{a,00} & 0 & 0 & e_{\overline{a}}+w \\
0 & \eta_{a,1}\widetilde{p}_{a,01} & e_{\overline{a}}-w & 0 \\
0 & e_{\overline{a}}-w & \eta_{a,0}\widetilde{p}_{a,10} & 0 \\
e_{\overline{a}}+w & 0 & 0 & \eta_{a,1}\widetilde{p}_{a,11} \\
\end{pmatrix},
\end{gathered}
\end{equation}
$w$ is an undefined parameter, and~$\overline{a}$ is an orthogonal basis to $a$. Further notations are introduced as
\begin{equation}
\begin{aligned} \label{notation1}
e_{\overline{a}} &= \frac{\eta_{a}}{4} \left(\widetilde{p}_{\overline{a},=} - \widetilde{p}_{\overline{a},\neq}\right), \qquad &\eta_a&=\sqrt{\eta_{a,0}\eta_{a,1}}, \\
\widetilde{p}_{a,=} &=\widetilde{p}_{a,00} + \widetilde{p}_{a,11}, \qquad &\widetilde{p}_{a,\neq} &=\widetilde{p}_{a,01} + \widetilde{p}_{a,10}, \\ 
\widetilde{p}_{a,ij} &= \frac{p_{a,ij}}{\eta_{a,j}}, \qquad &p_{a,\alpha}&=p_{a,\alpha \alpha}+p_{a,(1-\alpha)\alpha}, \\
p_{{\rm pass},a} &= p_{a,0}+p_{a,1}, \qquad &p_{{\rm pass}}&=\sum\limits_{a\in\{x,z\}}p_a^2p_{{\rm pass}, a}.
\end{aligned}
\end{equation}

Equation~(\ref{appendix_p_sift}) is an expression for the density matrix after the virtual measurement and sifting procedure. 
So, the~next step is to determine the key rate formulas for different approaches: the exact Equation~(\ref{KnewnewH}) and the simplified Equations~(\ref{K_new_case}) and (\ref{K_new_case_2}).

\subsection{Key Rate Equation~(\ref{KnewnewH})} \label{appendix_key_rate_general}
To derive Equation~(\ref{KnewnewH}) for the key rate we use Equation~(\ref{Knew}) with the supremum over matrices $\rho_{AB\widetilde{A}\widetilde{B}} \in \textbf{S}$, which have the form given by Equation~(\ref{appendix_p_sift}).

Equation~(\ref{tG}) requires the density matrix $\widetilde{\rho}_{X_{A}B \widetilde{\textbf{B}}}^{\rm sifted}$ after the decoherence in Alice's subsystem (see Equation~(\ref{rhoZX})). 
For convenience, in~order to express the values of Alice's register as $0/1$ rather than $+/-$, we additionally apply the Hadamard transform on it:
\vspace{-12pt}
\begin{adjustwidth}{-\extralength}{0cm}
\centering 
\begin{equation} \label{rhoAB2}
p_{\rm pass} \, \widetilde{\rho}_{X_{A}B \widetilde{\textbf{B}}}^{\rm sifted} \rightarrow  p_{\rm pass}H_A\mathcal{X}\left(\widetilde{\rho}
^{\rm sifted} \right)H_A 
= \sum\limits_{a\in\{x,z\}}\ket{a}_{\widetilde{\textbf{B}}}\!\bra{a} \otimes \sum\limits_{i=0}^1 \ket{i}_{A}\!\bra{i} \otimes p_a^2 \left(\textrm{D}_{0}^{a} + (-1)^{i}\textrm{D}_{1}^{a} \right),
\end{equation}
\end{adjustwidth}
where
\begin{equation} \label{D_01^a} 
\textrm{D}_{0}^{a} = \frac{1}{2}
\begin{pmatrix}
 p_{a,0} & 0 \\
 0 & p_{a,1} \\
\end{pmatrix}, \qquad
\textrm{D}_{1}^{a} =
\begin{pmatrix}
 0 &  e_{\overline{a}} \\
  e_{\overline{a}} & 0 \\
\end{pmatrix}.
\end{equation} 
So, the~density matrix of the \mbox{subsystem $B \widetilde{\textbf{B}}$} is given as
\begin{equation} \label{rhobb}
\begin{gathered}
\widetilde{\rho}_{B \widetilde{\textbf{B}}} = \Tr_{A} \widetilde{\rho}_{X_{A}B \widetilde{\textbf{B}}}^{\rm sifted} = \frac{2}{p_{\rm pass}} \sum\limits_{a\in\{x,z\}}\ket{a}_{\widetilde{B}}\!\bra{a} \otimes p_a^2\textrm{D}_{0}^{a}.
\end{gathered}
\end{equation}

Then we use the density matrices from Equations~(\ref{rhoAB2}) and (\ref{rhobb}) to calculate the entropy term in Equation~(\ref{Knew}).
The conditional entropy (see Equation~(\ref{HD})) can be written as follows:
\begin{equation} \label{hxbab}
H(X_{A}|B\widetilde{\textbf{B}})_{\widetilde{\rho}^{\rm sifted}} =   - \Tr \widetilde{\rho}_{X_{A}B \widetilde{\textbf{B}}}^{\rm sifted}\log_2 \widetilde{\rho}_{X_{A}B \widetilde{\textbf{B}}}^{\rm sifted} 
+ \Tr \widetilde{\rho}_{B \widetilde{\textbf{B}}} \log_2 \widetilde{\rho}_{B \widetilde{\textbf{B}}}. 
\end{equation}
Also, the~entropy of the subsystem $B \widetilde{\textbf{B}}$ has the following form:\vspace{-12pt}
\begin{adjustwidth}{-\extralength}{0cm}
\begin{equation} \label{hrhobb}
\begin{split}
H \left( \widetilde{\rho}_{B \widetilde{\textbf{B}}} \right) &= -\Tr \widetilde{\rho}_{B \widetilde{\textbf{B}}} \log_2 \widetilde{\rho}_{B \widetilde{\textbf{B}}}   
=\sum\limits_{a \in \{ x, z \}}\sum\limits_{\alpha = 0}^1 - \frac{p_a^2 p_{a,\alpha}}{p_{\rm pass}} \log_2 \frac{p_a^2 p_{a,\alpha}}{p_{\rm pass}} 
= \sum_{a\in\{x,z\}}
\frac{p_a^2p_{{\rm pass},a}}{p_{\rm pass}}h\!\left(\frac{p_{a,0}}{p_{{\rm pass}, a}}\right) 
+ h\!\left( \frac{p_z^2p_{{\rm pass},z}}{p_{{\rm pass}}}\right).
\end{split}
\end{equation}
\end{adjustwidth} 

Now, we consider the matrix $\widetilde{\rho}_{X_AB\widetilde{\textbf{B}}}^{\rm sifted}$ (see Equation~(\ref{rhoAB2})). 
It has the block-diagonal structure, and~we may deal with the basis independently. 
The eigenvalues of the matrix are
\begin{equation} \label{lambda1}
\begin{gathered}
\lambda_{\alpha}^{a} =\frac{p_a^2}{2p_{{\rm pass}}}\!\left\lbrace \frac{p_{{\rm pass},a}}{2} + (-1)^{\alpha} \sqrt{\left( \frac{p_{a,0}-p_{a,1}}{2} \right)^2 + 4e_{\overline{a}}^2} \right\rbrace.
\end{gathered}
\end{equation}

Note that the eigenvalues for $\textrm{D}_0^a+\textrm{D}_1^a$ and $\textrm{D}_0^a-\textrm{D}_1^a$ are equal. 
Thus, the~final key generation rate will be defined as Equation~(\ref{KnewnewH}), that is,\vspace{-6pt}
\begin{multline} \label{Knewnew1}
K \geq  {p_{\rm pass}} \sum\limits_{a \in \{ x, z \}}\sum\limits_{\alpha = 0}^1
\left(- \frac{p_a^2p_{a,\alpha}}{p_{\rm pass}}   \log_2 \frac{p_a^2p_{a,\alpha}}{p_{\rm pass}} +\lambda_{\alpha}^a \log_2 \lambda_{\alpha}^a \right) - p_{{\rm pass}}f h(Q) \\
= \sum\limits_{a \in \{ x, z \}} p_a^2p_{{\rm pass},a}\Bigg[h \left(\frac{1-\delta_{a,a}}{2}\right) - h \left(\frac{1-\sqrt{\delta_{a,a}^2+\delta_{a,\overline{a}}^2}}{2}\right) \Bigg] - p_{{\rm pass}}f h(Q),
\end{multline}
where
\begin{equation} \label{tau_xi}
\begin{gathered}
\delta_{a,a} = \frac{p_{a,0} - p_{a,1}}{p_{{\rm pass}, a}}, \qquad \delta_{a,\overline{a}} = \frac{\eta_{a} \left(\widetilde{p}_{\overline{a},=} - \widetilde{p}_{\overline{a},\neq}\right)}{p_{{\rm pass}, a}}.
\end{gathered}
\end{equation}

\subsection{Simplified Key Rate Equation~(\ref{K_new_case})} \label{denmat_1}
Consider the positive trace-preserving linear map
\begin{equation} \label{Phi_3}
\begin{gathered}
\Phi_3(\rho)= \frac{1}{2}\left(\rho + R_{\Phi,3} \rho R_{\Phi,3}^+\right), 
\\
R_{\Phi,3} = \left( 
\ket{z}_{\widetilde{A}}\!\bra{z} \otimes X_A + \ket{x}_{\widetilde{A}}\!\bra{x} \otimes Z_A \right)  \otimes \left(\ket{z}_{\widetilde{B}}\!\bra{z} \otimes X_B + \ket{x}_{\widetilde{B}}\!\bra{x} \otimes Z_B \right).
\end{gathered}
\end{equation}
This projector does not satisfy the commutation condition Equation~(\ref{ys}) of Proposition~\ref{Proposition1}. 
However, using the monotonicity of the quantum relative entropy in Equation~(\ref{HD}), we can estimate the key rate from below.
According to the fourth line of Equation~(\ref{entropy_ineq}), we can restrict the set of $\rho_{AB}$ to those that meet the equalities
\begin{equation}
\begin{gathered}
\left(X\otimes X \rho_{AB} X\otimes X \right) \otimes \ket{zz}_{\widetilde{A}\widetilde{B}}\!\bra{zz} =\rho_{AB} \otimes \ket{zz}_{\widetilde{A}\widetilde{B}}\!\bra{zz}, \\
\left(Z\otimes Z \rho_{AB} Z\otimes Z \right) \otimes \ket{xx}_{\widetilde{A}\widetilde{B}}\!\bra{xx} =\rho_{AB} \otimes \ket{xx}_{\widetilde{A}\widetilde{B}}\!\bra{xx}.
\end{gathered}
\end{equation}

Application of the map $\Phi_3$ to Equations~(\ref{rhoAB2}) and (\ref{D_01^a}) corresponds to the replacement of $p_{a,i}$ by
\begin{equation} \label{second_p}
    p_{a,0}=p_{a,1}=\frac{\eta_{a,0} + \eta_{a,1}}{4} = 2p'_{a}.
\end{equation}
 
Using the symmetries $\Phi_1-\Phi_3$, we obtain the density matrix after the sifting procedure and decoherence in Alice’s subsystem from Equation~(\ref{rhoZX}):
\begin{equation} \label{rho3_cond2}
p_{\rm pass} \mathcal{X}\left(\widetilde{\rho}
^{\rm sifted}\right)
\rightarrow \sum\limits_{a \in \{x,z\}}p_a^2
\ket{a}_{\widetilde{\textbf{B}}}\!\bra{a} \otimes {\rho'}^{a}_{AB},
\end{equation}
where
\begin{equation}  \label{rho3AB_cond2}
{\rho'}^{a}_{AB}= \begin{pmatrix}
p'_{a} & 0 & 0 & e_{\overline{a}} \\
0 & p'_{a} & e_{\overline{a}} & 0 \\
0 & e_{\overline{a}} & p'_{a} & 0 \\
e_{\overline{a}} & 0 & 0 & p'_{a} \\
\end{pmatrix}.
\end{equation}

Note that $\Phi_3$ is a trace-preserving map, so the values $p_{{\rm pass},a}$ and $p_{\rm pass}$ are preserved. 
Thus, according to the relationships (\ref{second_p}), we simplify the formula (\ref{KnewnewH}) for the key rate to the following form, which is exactly Equation~(\ref{K_new_case}):
\begin{equation} \label{Knewnew2}
K \geq 
\sum\limits_{a \in \{ x, z \}} p_a^2 p_{{\rm pass},a} \left( 1 -h \left(\frac{1-\Delta_a}{2} \right) \right) - \sum\limits_{a \in \{ x, z \}} {p_a^2} p_{{\rm pass},a} f h(Q),
\end{equation}
where
\begin{equation} 
\Delta_a = \frac{\eta_a\left( \widetilde{p}_{\overline{a},=} - \widetilde{p}_{\overline{a},\neq}\right)}{p_{{\rm pass},a}}.
\end{equation}
\subsection{Simplified Key Rate Equation~(\ref{K_new_case_2})} \label{denmat_2}
The derivation of Equation~(\ref{K_new_case_2}) is similar to that for Equation~(\ref{K_new_case}), but~with an additional symmetry applied.
In addition to $\Phi_1-\Phi_3$, we introduce the projector that aggregates statistics over different bases and has the following form:
\begin{equation}
\begin{gathered}
\Phi_4(\rho)= \frac{1}{2}\left(\rho + R_{\Phi,4} \rho R_{\Phi,4}^+\right), 
\\
R_{\Phi,4} = \left( 
 \ket{z}_{\widetilde{A}} \bra{x} \otimes \mathbb{H}_A + \ket{x}_{\widetilde{A}} \bra{z} \otimes \mathbb{H}_A \right) \otimes \left(\ket{z}_{\widetilde{B}} \bra{x} \otimes \mathbb{H}_B + \ket{x}_{\widetilde{B}} \bra{z} \otimes \mathbb{H}_B \right).
\end{gathered}
\end{equation}

Then, we restrict the set of $\rho_{AB}$ to those that meet the following
equalities:
\begin{equation}
\begin{gathered}
\left(\mathbb{H}\otimes \mathbb{H} \rho_{AB} \mathbb{H}\otimes \mathbb{H} \right) \otimes \ket{zz}_{\widetilde{A}\widetilde{B}}\!\bra{zz} =\rho_{AB} \otimes \ket{xx}_{\widetilde{A}\widetilde{B}}\!\bra{xx}, \\
\left(\mathbb{H}\otimes \mathbb{H} \rho_{AB} \mathbb{H}\otimes \mathbb{H} \right) \otimes \ket{xx}_{\widetilde{A}\widetilde{B}}\!\bra{xx} =\rho_{AB} \otimes \ket{zz}_{\widetilde{A}\widetilde{B}}\!\bra{zz}.
\end{gathered}
\end{equation}

Application of the maps $\Phi_3$ and $\Phi_4$ to Equations~(\ref{rhoAB2}) and (\ref{D_01^a}) corresponds to the replacement of $p_{a,i}$ and $e_a$ by

\begin{equation}
\begin{split}
p_{a,i}&=\frac{1}{8} \sum \limits_{a \in\{x,z\}} \sum \limits_{i = 0}^1 p_a^2\eta_{a,i} = 2p''_{\rm pass}, \\
e_{a}=e_{\overline{a}}&=\frac{1}{4} \sum \limits_{a \in\{x,z\}}p_a^2\eta_a \left( \widetilde{p}_{a,=} -\widetilde{p}_{a,\neq} \right)=2e''.
\end{split}
\end{equation}

Using this new symmetry, we have the following density matrix  from Equation~(\ref{rhoZX}):
\begin{equation}
p_{\rm pass} \mathcal{X}\left(\widetilde{\rho}
^{\rm sifted} \right)
\rightarrow \sum\limits_{a \in \{x,z\}}\ket{a}_{\widetilde{\textbf{B}}} \bra{a} \otimes \rho''_{AB},
\end{equation}
where
\begin{equation}
\rho''_{AB}= \begin{pmatrix}
p''_{\rm pass} & 0 & 0 & e'' \\
0 & p''_{\rm pass} & e'' & 0 \\
0 & e'' & p''_{\rm pass} & 0 \\
e'' & 0 & 0 & p''_{\rm pass} \\
\end{pmatrix}.
\end{equation}

As $\Phi_3$ and $\Phi_4$ are trace-preserving maps, the~value $p_{\rm pass}$ is preserved. As~a result, the~final bound for the key rate has the following form, which is exactly Equation~(\ref{K_new_case_2}):
\begin{equation} \label{Knewnew3}
K  \geq p_{\rm pass} \left( 1 - h\!\left(\frac{1-\Delta''}{2} \right) \right) -p_{\rm pass}f h(Q),
\end{equation}
where
\begin{equation} 
\Delta'' = \frac{e''}{p''_{\rm pass}} = \frac{1}{p_{\rm pass}}{\sum\limits_{a \in\{x,z\}} p_a^2\eta_a\left( \widetilde{p}_{\overline{a},=} - \widetilde{p}_{\overline{a},\neq}\right)}.
\end{equation}


\section{Finite Key Size~Effects}\label{FiniteKey} 
Here, we receive Equation~(\ref{key_finite_init}) and present a derivation for taking into account finite key size effects, i.e.,~statistical fluctuation for the observed values and collective attack (independent and identically distributed case). Note that the theory does not consider the general attack on the~protocol.

In this appendix, we assume that all measurements are independent and identically distributed. So we can find confidence limits for observed values Equation~(\ref{tG}) using Hoeffding's inequality and then use the results of Ref.~\cite{Scarani_2008} to consider the case where observed values do not lie in the confidence interval. After~that we work with the system on $M$ copies of the state $\widetilde{\rho}^{\rm sifted}$ Equation~(\ref{rhoAB2}) and apply the leftover hash lemma to receive the lower bound for the secret key~\cite{tan2026makingleftoversimprovedprivacy}. $M$ is the sifted key~length.

The first step is to estimate the bounds for values $p_{a,ij}$. The~sifted key length is $M$, the~measured values of pulse receive probability is $\hat{p}_{a,ij}$ and the confidence parameter $\delta$. Using Hoeffding's {inequality:} 
\begin{equation}
    \Pr\left(|{p}_{a,ij} - \hat{p}_{a,ij}| \geq \varepsilon\right) \leq 2e^{-2M\varepsilon^2} = \delta,
\end{equation}
\begin{equation}
    \varepsilon = \sqrt{\frac{\ln(2/\delta)}{2M}}.
\end{equation}

The maximum of the von Neumann entropy for the real state ${\widetilde{\rho}}^{\rm sifted}$ is
\begin{equation} 
\overline{H}(X_{A}|B\widetilde{\textbf{B}})_{\widetilde{\rho}^{\rm sifted}}=\sup\limits_{|\hat{p}_{a,ij} - p_{a,ij}| \leq \varepsilon} H(X_{A}|B\widetilde{\textbf{B}})_{\widetilde{\rho}^{\rm sifted}}.
\end{equation}

We can account for the case where the values do not lie in the confidence interval by using the results of Ref.~\cite{Scarani_2008}. So, we add the probability that the values do not belong to the confidence interval $\delta$ to the security parameter $\varepsilon_{\mathrm{sec}} = 10^{-10}$. So, we have $\varepsilon'_{\mathrm{sec}}=\varepsilon_{\mathrm{sec}}+\delta$.


To estimate the effects of collective attack, we use the leftover hash lemma~\cite{tan2026makingleftoversimprovedprivacy}:
\begin{equation} \label{key_finite}
\ell_{\mathrm{ext}}^{\varepsilon_{\mathrm{sec}}} \geq \min\limits_{0\leq\delta\leq\varepsilon_{\mathrm{sec}}} \left[H_{\min}^{\delta}(A|E\widetilde{\textbf{B}})_{(\hat{\widetilde{\rho}}^{\rm sifted})^{\otimes M}} - \log_2 \frac{1}{3(\varepsilon'_{\mathrm{sec}}-\delta^2)^2}\right] - l_{\textrm{EC}}.
\end{equation}
The minimum of entropy $H_{\min}^{\delta}(A|E\widetilde{\textbf{B}})_{(\hat{\widetilde{\rho}}^{\rm sifted})^{\otimes M}}$ is calculated by the asymptotic equipartition property~\cite{Dupuis2020}:
\begin{equation} \label{H_min}
\begin{gathered}
\frac{1}{M} H_{\min}^{\delta}(A|E\widetilde{\textbf{B}})_{(\hat{\widetilde{\rho}}^{\rm sifted})^{\otimes M}} \geq  H(A|E\widetilde{\textbf{B}})_{\hat{\widetilde\rho}^{\rm sifted}} - 2 \sqrt{\frac{1-2\log_2\delta}{M}}\log_2(5) 
\geq \\ 
1-\sup\limits_{\rho_{AB\widetilde{A}\widetilde{B}} \in \textbf{S}}\overline{H}(X_{A}|B\widetilde{\textbf{B}})_{\widetilde{\rho}^{\rm sifted}} - 2 \sqrt{\frac{1-2\log_2\delta}{M}}\log_2(5).
\end{gathered}
\end{equation}

Substituting {Ineqs.} (\ref{hxbab}) and (\ref{H_min}) in Equation~(\ref{key_finite}) for the secret key length, we obtain the final formula Equation~(\ref{key_finite_init}) for key that accounts for finite-size effects.

\section{\label{mono} Proof of the Key Rate Monotonicity over Detection Efficiency Mismatch~Parameters}
Let us prove the monotonicity of the key rates with respect to all detection efficiencies $\eta_{a,\alpha}$ given by Equations~(\ref{KnewnewH}), (\ref{K_new_case}), and~(\ref{K_new_case_2}).
This fact allows us to use the lower bounds for all $\eta_{a,\alpha}$ in practical calculations since the exact values $\eta_{a,\alpha}$ may not be~known.

The monotonicity of the key rates can be proved using the monotonicity of the quantum relative entropy in Equation~(\ref{HD}). 
Indeed, reducing $\eta_{a,\alpha}\to\theta\eta_{a,\alpha}$, where $a\in{\{z,x\}}$, $\alpha\in{0,1}$ and $\theta\in[0,1]$, is equivalent to the action of the additional map
\begin{equation}
\label{EqAddDecay}
    \rho_{AB\tilde A\widetilde B}\to
    \sum\limits_{a \in \{x,z\}}
    G_{a,\alpha}
    \rho_{AB\tilde A\widetilde B}
    G_{a,\alpha},
\end{equation}
where
\begin{equation}
    G_{z,\alpha}=I- \sum\limits_{\alpha=0}^1 \sqrt{1-\theta}\,
    \ket{\alpha}_{B}\!\bra{\alpha}
    \otimes
    \ket{a}_{\widetilde B}\!\bra a,  \quad 
    G_{x,\alpha}=I- \sum\limits_{\alpha=0}^1 \sqrt{1-\theta}\,
    H\ket{\alpha}_{B}\!\bra{\alpha}H
    \otimes
    \ket{a}_{\widetilde B}\!\bra a,    
\end{equation}
on Bob's subsystems $B\widetilde B$ in the arguments of the quantum relative entropy in Equation~(\ref{HD}). 
This map does not preserve the trace, but~the trace of  both arguments is reduced by the same factor. 
Thus, the~relative entropy is not increased under the action of this map on both arguments. 
Hence, given the same observable statistics $\Tr\Gamma_i\rho_{AB\widetilde A\widetilde B}$ (see Equation~(\ref{tG})), the~key rate given by Equation~(\ref{Knew}) does not increase. As~a consequence, the~key rates given by Equations~(\ref{KnewnewH}), (\ref{K_new_case}), and~(\ref{K_new_case_2}) are monotonic with respect to each $\eta_{a,\alpha}$.


\section{Estimation of the Key Rate by the Method of Ref.~\cite{Wang2025}}\label{key_rate_old}
Here, we calculate the key rate using Ref.~\cite{Wang2025} (Equation~{(B7)}).  
We employ the formula for the decoy-state BB84 passive basis choice protocol with memoryless detectors and adapt it for our case of the single-photon output and asymptotic regime. 
These changes lead to the removal of finite-size epsilon terms and multi-click terms. 
As a result, the~key rate is
\begin{eqnarray} \label{D1}
K  &=&
\max\left(\mathcal{B}_1 \left( 1 - h\left( \mathcal{B}_e\right)\right)-\lambda_{EC}, 0\right),
\\ 
\label{Be}
\mathcal{B}_e &=& \frac{\mathcal{B}_{1,\neq}^{\rm{decoy}}}{a\mathcal{B}_1} + \frac{\delta_{\max}}{a}\left( \frac{\mathcal{B}_{1}^{\rm{decoy}}}{\mathcal{B}_1} + 1\right),
\end{eqnarray}
where the notation is taken from Ref.~\cite{Wang2025}. 
Here, $\lambda_{EC}$ is the error correction term, 
$\mathcal{B}_1$ is an estimation of the count rate of the received signal single-photon states without dark counts,
$\mathcal{B}_e$ is an estimation of the error rate in the signal single-photon pulses, $\mathcal{B}_{1,\neq}^{\rm{decoy}}$, $\mathcal{B}_{1}^{\rm{decoy}}$ are an estimation
of receiving error rates and count rates in single-photon states by decoy-state method, 
$\delta_{\rm max}$ is a parameter describing the basis-efficiency mismatch. 
Also, $a$ is a parameter characterizing the beamsplitter:
\begin{equation} 
\begin{gathered}
a = \frac{p_x s}{p_z \left( 1 -s\right)},
\end{gathered}
\end{equation}
where $s= 1/2$ is the beam-splitting ratio on Bob's side, and~$p_x$ and $p_z$ are (as previous) the probabilities of choice X and Z bases by~Alice.

We will show that the error rate estimation given in Equation~(\ref{Be}) is large for our practical detection efficiency mismatch, which leads to a zero key rate. 
Namely, the~first term in Equation~(\ref{Be}) is of the order of the actual single-photon QBER and the second term is problematic.
The factor in the brackets lies in the interval $(1,2)$, but~$\delta_{\max}$ is large for our numerical~values.

In Equation~(\ref{Be}), $\delta_{\max}$ corresponds to the detection efficiency mismatch between different bases. It is derived in Ref.~\cite{Wang2025} ({Appendix~D}) and has the following form:
\begin{equation}
\delta_{\max} = (1 + \sqrt{a})\frac{p_x \zeta k^2}{2} 
\max \left\{ \sqrt{\frac{p_x u_x + p_z u_z}{p_x l_z}},\; \sqrt{\frac{p_x u_x + p_z u_z}{a p_z l_x}} \right\},
\end{equation}
where $\zeta$ is the normalized difference between clicks in different bases, 
$k^2$ corresponds to the normalization of the total number of received clicks, 
and $l_{x}, l_{z}, u_{x}, u_{z}$ correspond to the lower and upper bounds for click probabilities in the bases X and Z. 
These parameters are provided by the following equations:
\begin{equation}
\begin{split}
l_{z}&= (1 - s)\eta_{(z,\mathrm{min})}(1 - d)^{3} +\left((1 - s)(1 - \eta_{(z,\mathrm{max})}) + s(1 - \eta_{(x,\mathrm{max})})\right)  2d (1 - d)^{3} , \\
l_{x}&= s\eta_{(x,\mathrm{min})}(1 - d)^{3}  + \left((1 - s)(1 - \eta_{(z,\mathrm{max})}) + s(1 - \eta_{(x,\mathrm{max})})\right) 2d (1 - d)^{3}, \\
u_{z}&= (1 - s)\eta_{(z,\mathrm{max})}(1 - d)^{3} + \left((1 - s)(1 - \eta_{(z,\mathrm{min})}) + s(1 - \eta_{(x,\mathrm{min})})\right)  2d (1 - d)^{3}, \\
u_{x}&= s \eta_{(x,\mathrm{max})}(1 - d)^{3} +\left((1 - s)(1 - \eta_{(z,\mathrm{min})}) + s(1 - \eta_{(x,\mathrm{min})})\right) 2d (1 - d)^{3},
\end{split}
\end{equation}
\begin{equation} 
k^2 = 
\begin{cases}
\left(p_x l_x + p_z l_z\right)^{-1} & \text{if } p_x l_x + p_z l_z > 0 \\
\infty & \text{if } p_x l_x + p_z l_z \leq 0
\end{cases}
\end{equation}
\begin{equation} 
\begin{gathered}
\zeta =\Delta_{1,\max} + \Delta_{2,\max}, \\
 \Delta_{1,\max} = s(1 - d)^3  \left[ 2\eta_{\max} - \eta_{(x,\min)} - \eta_{(z,\min)} \right], 
  \\  
 \Delta_{2,\max} =
\frac{1}{1 - s}(1 - \eta_{\min}) 2d (1 - d)^3,
\end{gathered}
\end{equation}
where $d = Y_0/4 = 7.5\times10^{-7}$ is the dark count rate, 
$\eta_{(b,{\mathrm{max}})}, \eta_{(b,{\mathrm{min}})}$ are the maximum and minimum efficiencies in the corresponding bases $b \in\{x,z\}$, 
and $\eta_{\max}, \eta_{\min}$ are the maximum and minimum efficiencies among~all.

Combining all the equations above, we can derive the probability of errors in single-photon states for our parameters from Tables~ \ref{table:station} and \ref{table:param}. The~upper bound for the probability of errors $\mathcal{B}_e$ in Equation~(\ref{D1}) is around 1.5, i.e.,~it is trivial, so the final key rate is zero.


\reftitle{References}


\bibliography{sample}

@article{Bennett2014,
author = {Charles H. Bennett and Gilles Brassard},
title = {Quantum cryptography: Public key distribution and coin tossing},
journal = {Theoretical Computer Science},
volume = {560},
pages = {7--11},
year = {2014},
issn = {0304-3975},
doi = {10.1016/j.tcs.2014.05.025},
url = {https://www.sciencedirect.com/science/article/pii/S0304397514004241},
note={(First publication: 1984)}
}

@article{Mayers2001,
author = {Mayers, Dominic},
title = {Unconditional security in quantum cryptography},
year = {2001},
issue_date = {May 2001},
publisher = {Association for Computing Machinery},
address = {New York, NY, USA},
volume = {48},
number = {3},
issn = {0004-5411},
url = {https://doi.org/10.1145/382780.382781},
doi = {10.1145/382780.382781},
journal = {Journal of the ACM},
month = {May},
pages = {351--406},
numpages = {56}
}

@article{Shor2000,
  title = {Simple Proof of Security of the {BB}84 Quantum Key Distribution Protocol},
  author = {Shor, Peter W. and Preskill, John},
  journal = {Physical Review Letters},
  volume = {85},
  issue = {2},
  pages = {441--444},
  numpages = {0},
  year = {2000},
  month = {Jul},
  publisher = {American Physical Society},
  doi = {10.1103/PhysRevLett.85.441},
  url = {https://link.aps.org/doi/10.1103/PhysRevLett.85.441}
}

@article{Renner2008,
author = {Renner, Renato},
title = {Security of quantum key distribution},
journal = {International Journal of Quantum Information},
volume = {6},
number = {01},
pages = {1--127},
year = {2008},
doi = {10.1142/S0219749908003256},
URL = {https://doi.org/10.1142/S0219749908003256},
}

@article{Koashi2009,
doi = {10.1088/1367-2630/11/4/045018},
url = {https://dx.doi.org/10.1088/1367-2630/11/4/045018},
year = {2009},
month = {apr},
volume = {11},
number = {4},
pages = {045018},
author = {Koashi, M},
title = {Simple security proof of quantum key distribution based on complementarity},
journal = {New Journal of Physics}
}

@article{Tomamichel2017,
  doi = {10.22331/q-2017-07-14-14},
  url = {https://doi.org/10.22331/q-2017-07-14-14},
  title = {A largely self-contained and complete security proof for quantum key  distribution},
  author = {Tomamichel, Marco and Leverrier, Anthony},
  journal = {Quantum},
  issn = {2521-327X},
  volume = {1},
  pages = {14},
  year = {2017}
}

@article{Gisin2002,
  title = {Quantum cryptography},
  author = {Gisin, Nicolas and Ribordy, Gr\'egoire and Tittel, Wolfgang and Zbinden, Hugo},
  journal = {Review of Modern Physics},
  volume = {74},
  issue = {1},
  pages = {145--195},
  numpages = {0},
  year = {2002},
  month = {Mar},
  publisher = {American Physical Society},
  doi = {10.1103/RevModPhys.74.145},
  url = {https://link.aps.org/doi/10.1103/RevModPhys.74.145}
}

@article{Xu2020,
  title = {Secure quantum key distribution with realistic devices},
  author = {Xu, Feihu and Ma, Xiongfeng and Zhang, Qiang and Lo, Hoi-Kwong and Pan, Jian-Wei},
  journal = {Review of Modern Physics},
  volume = {92},
  issue = {2},
  pages = {025002},
  numpages = {60},
  year = {2020},
  month = {May},
  publisher = {American Physical Society},
  doi = {10.1103/RevModPhys.92.025002},
  url = {https://link.aps.org/doi/10.1103/RevModPhys.92.025002}
}

@article{Diamanti2016,
  title = {Practical challenges in quantum key distribution},
  author = {Diamanti, Eleni and Lo, Hoi-Kwong and Qi, Bing and Yuan, Zhiliang},
  journal = {npj Quantum Information},
  volume = {2},
  issue = {1},
  pages = {16025},
  numpages = {60},
  year = {2016},
  month = {November},
  doi = {10.1038/npjqi.2016.25}
}

@article{Reutov2023,
 author = {Reutov, Aleksei and Tayduganov, Andrey and Mayboroda, Vladimir and Fat’yanov, Oleg},
 title = {Security of the Decoy-State {BB}84 Protocol with Imperfect State Preparation},
 journal = {Entropy},
 volume = {25},
 year = {2023},
 pages = {1556},
 number = {11},
 doi = {10.3390/e25111556}
 }

@article{Makarov2006,
  title = {Effects of detector efficiency mismatch on security of quantum cryptosystems},
  author = {Makarov, Vadim and Anisimov, Andrey and Skaar, Johannes},
  journal = {Physical Review A},
  volume = {74},
  issue = {2},
  pages = {022313},
  numpages = {11},
  year = {2006},
  month = {Aug},
  publisher = {American Physical Society},
  doi = {10.1103/PhysRevA.74.022313},
  url = {https://link.aps.org/doi/10.1103/PhysRevA.74.022313}
}

@article{Jain2016,
author = {Nitin Jain and Birgit Stiller and Imran Khan and Dominique Elser and Christoph Marquardt and Gerd Leuchs},
title = {Attacks on practical quantum key distribution systems (and how to prevent them)},
journal = {Contemporary Physics},
volume = {57},
number = {3},
pages = {366--387},
year = {2016},
publisher = {Taylor \& Francis},
doi = {10.1080/00107514.2016.1148333},
}

@article{Marcomini2025,
doi = {10.1088/2058-9565/adc8cc},
url = {https://dx.doi.org/10.1088/2058-9565/adc8cc},
year = {2025},
month = {apr},
publisher = {IOP Publishing},
volume = {10},
number = {3},
pages = {035002},
author = {Marcomini, Alessandro and Mizutani, Akihiro and Grünenfelder, Fadri and Curty, Marcos and Tamaki, Kiyoshi},
title = {Loss-tolerant quantum key distribution with detection efficiency mismatch},
journal = {Quantum Science and Technology}
}

@article{Lo2012,
  title = {Measurement-Device-Independent Quantum Key Distribution},
  author = {Lo, Hoi-Kwong and Curty, Marcos and Qi, Bing},
  journal = {Physical Review Letters},
  volume = {108},
  issue = {13},
  pages = {130503},
  numpages = {5},
  year = {2012},
  month = {Mar},
  publisher = {American Physical Society},
  doi = {10.1103/PhysRevLett.108.130503},
  url = {https://link.aps.org/doi/10.1103/PhysRevLett.108.130503}
}

@article{Tamaki2012,
  title = {Phase encoding schemes for measurement-device-independent quantum key distribution with basis-dependent flaw},
  author = {Tamaki, Kiyoshi and Lo, Hoi-Kwong and Fung, Chi-Hang Fred and Qi, Bing},
  journal = {Physical Review A},
  volume = {85},
  issue = {4},
  pages = {042307},
  numpages = {14},
  year = {2012},
  month = {Apr},
  publisher = {American Physical Society},
  doi = {10.1103/PhysRevA.85.042307},
  url = {https://link.aps.org/doi/10.1103/PhysRevA.85.042307}
}

@article{Bochkov2019,
  title = {Security of quantum key distribution with detection-efficiency mismatch in the single-photon case: Tight bounds},
  author = {Bochkov, M. K. and Trushechkin, A. S.},
  journal = {Physical Review A},
  volume = {99},
  issue = {3},
  pages = {032308},
  numpages = {15},
  year = {2019},
  month = {Mar},
  publisher = {American Physical Society},
  doi = {10.1103/PhysRevA.99.032308},
  url = {https://link.aps.org/doi/10.1103/PhysRevA.99.032308}
}

@article{Trushechkin2022,
  doi = {10.22331/q-2022-07-22-771},
  url = {https://doi.org/10.22331/q-2022-07-22-771},
  title = {Security of quantum key distribution with detection-efficiency mismatch in the multiphoton case},
  author = {Trushechkin, Anton},
  journal = {Quantum},
  issn = {2521-327X},
  volume = {6},
  pages = {771},
  year = {2022}
}

@article{Tupkary2025,
  doi = {10.22331/q-2025-12-11-1937},
  url = {https://doi.org/10.22331/q-2025-12-11-1937},
  title = {Phase error rate estimation in {QKD} with imperfect detectors},
  author = {Tupkary, Devashish and Nahar, Shlok and Sinha, Pulkit and L{\"{u}}tkenhaus, Norbert},
  journal = {{Quantum}},
  issn = {2521-327X},
  volume = {9},
  pages = {1937},
  year = {2025}
}

@article{Grasselli2025,
  title = {Quantum key distribution with basis-dependent detection probability},
  author = {Grasselli, Federico and Chesi, Giovanni and Walk, Nathan and Kampermann, Hermann and Widomski, Adam and Ogrodnik, Maciej and Karpi\ifmmode \acute{n}\else \'{n}\fi{}ski, Micha\l{} and Macchiavello, Chiara and Bru\ss{}, Dagmar and Wyderka, Nikolai},
  journal = {Physical Review Applied},
  volume = {23},
  issue = {4},
  pages = {044011},
  numpages = {85},
  year = {2025},
  month = {Apr},
  publisher = {American Physical Society},
  doi = {10.1103/PhysRevApplied.23.044011},
  url = {https://link.aps.org/doi/10.1103/PhysRevApplied.23.044011}
}

@article{Winick2018,
  doi = {10.22331/q-2018-07-26-77},
  url = {https://doi.org/10.22331/q-2018-07-26-77},
  title = {Reliable numerical key rates for quantum key distribution},
  author = {Winick, Adam and L{\"{u}}tkenhaus, Norbert and Coles, Patrick J.},
  journal = {{Quantum}},
  issn = {2521-327X},
  volume = {2},
  pages = {77},
  year = {2018}
}

@article{Coles2016,
  doi = {10.1038/ncomms11712},
  url = {https://doi.org/10.1038/ncomms11712},
  title = {Numerical approach for unstructured quantum key distribution},
  author = {Coles, Patrick J. and Metodiev, Eric M. and L{\"{u}}tkenhaus, Norbert},
  journal = {Nature Communications},
  issn = {1},
  volume = {7},
  pages = {11712},
  month = {May},
  year = {2016}
}

@article{Zhang2021,
  title = {Security proof of practical quantum key distribution with detection-efficiency mismatch},
  author = {Zhang, Yanbao and Coles, Patrick J. and Winick, Adam and Lin, Jie and L\"utkenhaus, Norbert},
  journal = {Physical Review Research},
  volume = {3},
  issue = {1},
  pages = {013076},
  numpages = {13},
  year = {2021},
  month = {Jan},
  publisher = {American Physical Society},
  doi = {10.1103/PhysRevResearch.3.013076},
  url = {https://link.aps.org/doi/10.1103/PhysRevResearch.3.013076}
}

@article{Nahar2026,
  doi = {10.22331/q-2026-03-24-2044},
  url = {https://doi.org/10.22331/q-2026-03-24-2044},
  title = {Imperfect detectors for adversarial tasks with applications to quantum key distribution},
  author = {Nahar, Shlok and Tupkary, Devashish and L{\"{u}}tkenhaus, Norbert},
  journal = {{Quantum}},
  issn = {2521-327X},
  volume = {10},
  pages = {2044},
  year = {2026}
}

@article{Liao2017,
	author = {Liao, Sheng-Kai and Cai, Wen-Qi and Liu, Wei-Yue and Zhang, Liang and Li, Yang and Ren, Ji-Gang and Yin, Juan and Shen, Qi and Cao, Yuan and Li, Zheng-Ping and Li, Feng-Zhi and Chen, Xia-Wei and Sun, Li-Hua and Jia, Jian-Jun and Wu, Jin-Cai and Jiang, Xiao-Jun and Wang, Jian-Feng and Huang, Yong-Mei and Wang, Qiang and Zhou, Yi-Lin and Deng, Lei and Xi, Tao and Ma, Lu and Hu, Tai and Zhang, Qiang and Chen, Yu-Ao and Liu, Nai-Le and Wang, Xiang-Bin and Zhu, Zhen-Cai and Lu, Chao-Yang and Shu, Rong and Peng, Cheng-Zhi and Wang, Jian-Yu and Pan, Jian-Wei},
	date = {2017/09/01},
	doi = {10.1038/nature23655},
	id = {Liao2017},
	isbn = {1476-4687},
	journal = {Nature},
	number = {7670},
	pages = {43--47},
	title = {Satellite-to-ground quantum key distribution},
	url = {https://doi.org/10.1038/nature23655},
	volume = {549},
	year = {2017}
    }

@article{Li2025,
  title={Microsatellite-based real-time quantum key distribution},
  author={Li, Yang and Cai, Wen-Qi and Ren, Ji-Gang and Wang, Chao-Ze and Yang, Meng and Zhang, Liang and Wu, Hui-Ying Singer and Chang, Liang and Wu, Jin-Cai and Jin, Biao and Li, Xue-Jiao and Liu, Hui and Yu, Guang-Wen and Tao, Xue-Ying and Chen, Ting and Liu, Chong-Fei and Luo, Wen-Bin and Zhou, Jie and Yong, Hai-Lin and Li, Yu-Huai and Li, Feng-Zhi and Jiang, Cong and Chen, Hao-Ze and Wu, Chao and Tong, Xin-Hai Ed and Xie, Si-Jiang and Zhou, Fei and Liu, Wei-Yue and Ismail, Yaseera and Petruccione, Francesco and Liu, Nai-Le and Li, Li and Xu, Feihu and Cao, Yuan and Yin, Juan and Shu, Rong and Wang, Xiang-Bin and Zhang, Qiang and Wang, Jian-Yu and Liao, Sheng-Kai and Peng, Cheng-Zhi and Pan, Jian-Wei},
  journal={Nature},
  volume={640},
  number={8057},
  pages={47--54},
  year={2025},
  publisher={Nature Publishing Group UK London},
  doi = {10.1038/s41586-025-08739-z},
}

@misc{Wang2025,
author = {Wang, Zhiyao and Tupkary, Devashish and Nahar, Shlok},
title = {Phase error estimation for passive detection setups with imperfections and memory effects},
year = {2025},
archivePrefix={arXiv},
eprint={2508.21486},
primaryClass={quant-ph},
}

@article{Khmelev2024,
author = {Aleksandr Khmelev and Alexey Duplinsky and Ruslan Bakhshaliev and Egor Ivchenko and Liubov Pismeniuk and Vladimir Mayboroda and Ivan Nesterov and Arkadiy Chernov and Anton Trushechkin and Evgeniy Kiktenko and Vladimir Kurochkin and Aleksey Fedorov},
journal = {Opt. Express},
number = {7},
pages = {11964--11978},
publisher = {Optica Publishing Group},
title = {Eurasian-scale experimental satellite-based quantum key distribution with detector efficiency mismatch analysis},
volume = {32},
month = {Mar},
year = {2024},
url = {https://opg.optica.org/oe/abstract.cfm?URI=oe-32-7-11964},
doi = {10.1364/OE.511772}
}

@article{Ivchenko2025,
  title = {Secrecy of quantum key distribution in case passive basis choice with detection efficiency mismatch},
  author = {E.I. Ivchenko and A.S. Trushechkin and A.V. Khmelev and V.L. Kurochkin},
  journal = {J. Opt. Technol.},
  publisher = {Optica Publishing Group},
  volume = {92},
  issue = {7},
  pages = {468--472},
  year = {2025},
  doi = {10.1364/JOT.92.000468},
}

@article{Lu_2022,
  title = {Micius quantum experiments in space},
  author = {Lu, Chao-Yang and Cao, Yuan and Peng, Cheng-Zhi and Pan, Jian-Wei},
  journal = {Rev. Mod. Phys.},
  volume = {94},
  issue = {3},
  pages = {035001},
  numpages = {46},
  year = {2022},
  month = {Jul},
  publisher = {American Physical Society},
  doi = {10.1103/RevModPhys.94.035001},
  url = {https://link.aps.org/doi/10.1103/RevModPhys.94.035001}
}

@article{Goswami_25,
author = {Sumit Goswami and Sayandip Dhara and Neil Sinclair and Makan Mohageg and Jasminder S. Sidhu and Sabyasachi Mukhopadhyay and Markus Krutzik and John R. Lowell and Daniel K. L. Oi and Mustafa G\"{u}ndo\u{g}an and Ying-Cheng Chen and Hsiang-Hua Jen and Christoph Simon},
journal = {Optica Quantum},
number = {6},
pages = {590--605},
publisher = {Optica Publishing Group},
title = {Satellites promise global-scale quantum networks},
volume = {3},
month = {Dec},
year = {2025},
doi = {10.1364/OPTICAQ.567531},
}

@misc{Rozenman2026,
  title         = {Free-space and Satellite-Based Quantum Communication: Principles, Implementations, and Challenges}, 
  author        = {Georgi Gary Rozenman and Alona Maslennikov and Sara P. Gandelman and Yuval Reches and Sahar Delfan and Neel Kanth Kundu and Leyi Zhang and Ruiqi Liu},
  year          = {2026},
  eprint        = {2602.01426},
  archivePrefix = {arXiv},
  primaryClass  = {quant-ph}
}

@misc{Alnaseri2026,
  title         = {A Review on Quantum Satellite Communications: Challenges and Future Directions}, 
  author        = {Omar Alnaseri and Yassine Himeur and Ahmed Al Asadi and Mohammed A. Ala'anzy and Rida Gadhafi and Shadi Atalla and Wathiq Mansoor},
  year          = {2026},
  eprint        = {2605.02272},
  archivePrefix = {arXiv},
  primaryClass  = {eess.SP}
}

@ARTICLE{Aljebry2026,
  author={Aljebry, Dalal F. and Barnawi, Ahmed},
  journal={IEEE Access}, 
  title={Quantum Key Distribution-as-a-Service ({QKD}aa{S}) for Smart Megacities: A Satellite-Based Framework and Performance Benchmarking}, 
  year={2026},
  volume={14},
  number={},
  pages={121793-121814},
  doi={10.1109/ACCESS.2026.3716359}
  }

@article{Orsucci2025,
author = {Orsucci, Davide and Kleinpaß, Philipp and Meister, Jaspar and De Marco, Innocenzo and Häusler, Stefanie and Strang, Thomas and Walenta, Nino and Moll, Florian},
title = {Assessment of Practical Satellite Quantum Key Distribution Architectures for Current and Near-Future Missions},
journal = {International Journal of Satellite Communications and Networking},
volume = {43},
number = {3},
pages = {164-192},
doi = {https://doi.org/10.1002/sat.1544},
year = {2025}
}

@article{Zhang2017_DEM,
  title = {Entanglement verification with detection-efficiency mismatch},
  author = {Zhang, Yanbao and L\"utkenhaus, Norbert},
  journal = {Physical Review A},
  volume = {95},
  issue = {4},
  pages = {042319},
  numpages = {20},
  year = {2017},
  month = {Apr},
  publisher = {American Physical Society},
  doi = {10.1103/PhysRevA.95.042319},
  url = {https://link.aps.org/doi/10.1103/PhysRevA.95.042319}
}

@article{Devetak2005,
  title = {Distillation of secret key and entanglement from quantum states},
  author = {Devetak, Igor and Winter, Andreas},
  journal = {Proceedings of the Royal Society A},
  volume = {461},
  issue = {2053},
  pages = {207},
  year = {2005},
  month = {Jan},
  publisher = {Royal Society},
  doi = {10.1098/rspa.2004.1372},
  url = {https://doi.org/10.1098/rspa.2004.1372}
}

@article{Berta2010,
  title = {The uncertainty principle in the presence of quantum memory},
  author = {Berta, Mario and Christandl, Matthias and Colbeck, Roger and Renes, Joseph M. and Renner, Renato},
  journal = {Nature Physics},
  volume = {6},
  issue = {9},
  pages = {659},
  year = {2010},
  month = {Sep},
  doi = {10.1038/nphys1734},
  url = {https://doi.org/10.1038/nphys1734}
}

@article{Coles2011,
  title = {Information-theoretic treatment of tripartite systems and quantum channels},
  author = {Coles, Patrick J. and Yu, Li and Gheorghiu, Vlad and Griffiths, Robert B.},
  journal = {Physical Review A},
  volume = {83},
  issue = {6},
  pages = {062338},
  numpages = {19},
  year = {2011},
  month = {Jun},
  publisher = {American Physical Society},
  doi = {10.1103/PhysRevA.83.062338},
  url = {https://link.aps.org/doi/10.1103/PhysRevA.83.062338}
}

@article{Dusek2000,
  title = {Unambiguous state discrimination in quantum cryptography with weak coherent states},
  author = {Du\ifmmode \check{s}\else \v{s}\fi{}ek, Miloslav and Jahma, Mika and L\"utkenhaus, Norbert},
  journal = {Physical Review A},
  volume = {62},
  issue = {2},
  pages = {022306},
  numpages = {9},
  year = {2000},
  month = {Jul},
  publisher = {American Physical Society},
  doi = {10.1103/PhysRevA.62.022306},
  url = {https://link.aps.org/doi/10.1103/PhysRevA.62.022306}
}

@article{Lutkenhaus_2002,
doi = {10.1088/1367-2630/4/1/344},
url = {https://dx.doi.org/10.1088/1367-2630/4/1/344},
year = {2002},
month = {jul},
publisher = {},
volume = {4},
number = {1},
pages = {44},
author = {Lütkenhaus, Norbert and Jahma, Mika},
title = {Quantum key distribution with realistic states:
photon-number statistics  in the photon-number splitting attack},
journal = {New Journal of Physics}
}

@article{Wang2005,
  title = {Beating the Photon-Number-Splitting Attack in Practical Quantum Cryptography},
  author = {Wang, Xiang-Bin},
  journal = {Physical Review Letters},
  volume = {94},
  issue = {23},
  pages = {230503},
  numpages = {4},
  year = {2005},
  month = {Jun},
  publisher = {American Physical Society},
  doi = {10.1103/PhysRevLett.94.230503},
  url = {https://link.aps.org/doi/10.1103/PhysRevLett.94.230503}
}

@article{Lo2005,
  title = {Decoy State Quantum Key Distribution},
  author = {Lo, Hoi-Kwong and Ma, Xiongfeng and Chen, Kai},
  journal = {Physical Review Letters},
  volume = {94},
  issue = {23},
  pages = {230504},
  numpages = {4},
  year = {2005},
  month = {Jun},
  publisher = {American Physical Society},
  doi = {10.1103/PhysRevLett.94.230504},
  url = {https://link.aps.org/doi/10.1103/PhysRevLett.94.230504}
}

@article{Ma2005,
  title = {Practical decoy state for quantum key distribution},
  author = {Ma, Xiongfeng and Qi, Bing and Zhao, Yi and Lo, Hoi-Kwong},
  journal = {Physical Review A},
  volume = {72},
  issue = {1},
  pages = {012326},
  numpages = {15},
  year = {2005},
  month = {Jul},
  publisher = {American Physical Society},
  doi = {10.1103/PhysRevA.72.012326},
  url = {https://link.aps.org/doi/10.1103/PhysRevA.72.012326}
}

@article{Zhang2017_decoy,
  title = {Improved key-rate bounds for practical decoy-state quantum-key-distribution systems},
  author = {Zhang, Zhen and Zhao, Qi and Razavi, Mohsen and Ma, Xiongfeng},
  journal = {Physical Review A},
  volume = {95},
  issue = {1},
  pages = {012333},
  numpages = {14},
  year = {2017},
  month = {Jan},
  publisher = {American Physical Society},
  doi = {10.1103/PhysRevA.95.012333},
  url = {https://link.aps.org/doi/10.1103/PhysRevA.95.012333}
}

@article{Khmelev2023_model,
AUTHOR = {Khmelev, Aleksandr V. and Ivchenko, Egor I. and Miller, Alexander V. and Duplinsky, Alexey V. and Kurochkin, Vladimir L. and Kurochkin, Yury V.},
TITLE = {Semi-Empirical Satellite-to-Ground Quantum Key Distribution Model for Realistic Receivers},
JOURNAL = {Entropy},
VOLUME = {25},
YEAR = {2023},
NUMBER = {4},
pages = {670},
URL = {https://www.mdpi.com/1099-4300/25/4/670},
PubMedID = {37190458},
ISSN = {1099-4300},
doi = {10.3390/e25040670}
}

@article{Zhao2008,
  title = {Quantum hacking: Experimental demonstration of time-shift attack against practical quantum-key-distribution systems},
  author = {Zhao, Yi and Fung, Chi-Hang Fred and Qi, Bing and Chen, Christine and Lo, Hoi-Kwong},
  journal = {Physical Review A},
  volume = {78},
  issue = {4},
  pages = {042333},
  numpages = {5},
  year = {2008},
  month = {Oct},
  publisher = {American Physical Society},
  doi = {10.1103/PhysRevA.78.042333},
  url = {https://link.aps.org/doi/10.1103/PhysRevA.78.042333}
}

@article{Sajeed2015,
  title = {Security loophole in free-space quantum key distribution due to spatial-mode detector-efficiency mismatch},
  author = {Sajeed, Shihan and Chaiwongkhot, Poompong and Bourgoin, Jean-Philippe and Jennewein, Thomas and L\"utkenhaus, Norbert and Makarov, Vadim},
  journal = {Physical Review A},
  volume = {91},
  issue = {6},
  pages = {062301},
  numpages = {6},
  year = {2015},
  month = {Jun},
  publisher = {American Physical Society},
  doi = {10.1103/PhysRevA.91.062301},
  url = {https://link.aps.org/doi/10.1103/PhysRevA.91.062301}
}

@article{Pirandola2020,
author = {S. Pirandola and U. L. Andersen and L. Banchi and M. Berta and D. Bunandar and R. Colbeck and D. Englund and T. Gehring and C. Lupo and C. Ottaviani and J. L. Pereira and M. Razavi and J. Shamsul Shaari and M. Tomamichel and V. C. Usenko and G. Vallone and P. Villoresi and P. Wallden},
journal = {Adv. Opt. Photon.},
number = {4},
pages = {1012--1236},
publisher = {Optica Publishing Group},
title = {Advances in quantum cryptography},
volume = {12},
month = {Dec},
year = {2020},
url = {https://opg.optica.org/aop/abstract.cfm?URI=aop-12-4-1012},
doi = {10.1364/AOP.361502}
}

@article{Muller_Hermes2017,
  title = {Monotonicity of the Quantum Relative Entropy Under Positive Maps},
  author = {Müller-Hermes, Alexander and Reeb, David},
  journal = {Annales Henri Poincar\'e},
  volume = {18},
  issue = {5},
  pages = {1777},
  year = {2017},
  month = {May},
  publisher = {Springer},
  doi = {10.1007/s00023-017-0550-9},
  url = {https://doi.org/10.1007/s00023-017-0550-9}
}

@article{Takenaka2017,
  author  = {Takenaka, Hideki and Carrasco-Casado, Alberto and Fujiwara, Mikio and Kitamura, Mitsuo and Sasaki, Masahide and Toyoshima, Morio},
  title   = {Satellite-to-ground quantum-limited communication using a 50-kg-class microsatellite},
  journal = {Nature Photonics},
  year    = {2017},
  volume  = {11},
  number  = {8},
  pages   = {502--508},
  month   = {Aug},
  doi     = {10.1038/nphoton.2017.107}
}

@INPROCEEDINGS{Bacsardi2025,
  author={Bacsardi, Laszlo and Mihaly, Andras and Olah, Kitti},
  booktitle={2025 International Conference on Quantum Communications, Networking, and Computing (QCNC)}, 
  title={From free-space quantum communication systems to satellite-based networks}, 
  year={2025},
  volume={},
  number={},
  pages={100-104},
  doi={10.1109/QCNC64685.2025.00023}
  }

@article{Ekert91,
  title = {Quantum cryptography based on Bell's theorem},
  author = {Ekert, Artur K.},
  journal = {Phys. Rev. Lett.},
  volume = {67},
  issue = {6},
  pages = {661--663},
  numpages = {0},
  year = {1991},
  month = {Aug},
  publisher = {American Physical Society},
  doi = {10.1103/PhysRevLett.67.661},
  url = {https://link.aps.org/doi/10.1103/PhysRevLett.67.661}
}

@misc{tan2026makingleftoversimprovedprivacy,
      title={Making the most of leftovers: Improved privacy amplification for quantum key distribution}, 
      author={Matthew Simon Tan and Bartosz Regula and Marco Tomamichel},
      year={2026},
      eprint={2609.38306},
      archivePrefix={arXiv},
      primaryClass={quant-ph},
      url={https://arxiv.org/abs/2609.38306}, 
}

@article{Dupuis2020,
  title = {Entropy Accumulation},
  author = {Dupuis, Fr\'ed\'eric and Fawzi, Omar and Renner, Renato},
  journal = {Communications in Mathematical Physics},
  volume = {379},
  number = {3},
  pages = {867--913},
  year = {2020},
  month = {Nov},
  publisher = {Springer Science and Business Media LLC},
  doi = {10.1007/s00220-020-03839-5},
  url = {https://doi.org/10.1007/s00220-020-03839-5}
}

@article{BBM92,
  title = {Quantum cryptography without Bell's theorem},
  author = {Bennett, Charles H. and Brassard, Gilles and Mermin, N. David},
  journal = {Phys. Rev. Lett.},
  volume = {68},
  issue = {5},
  pages = {557--559},
  numpages = {0},
  year = {1992},
  month = {Feb},
  publisher = {American Physical Society},
  doi = {10.1103/PhysRevLett.68.557},
  url = {https://link.aps.org/doi/10.1103/PhysRevLett.68.557}
}

@article{Scarani_2008,
   title={Quantum Cryptography with Finite Resources: Unconditional Security Bound for Discrete-Variable Protocols with One-Way Postprocessing},
   volume={100},
   ISSN={1079-7114},
   url={http://dx.doi.org/10.1103/PhysRevLett.100.200501},
   DOI={10.1103/physrevlett.100.200501},
   number={20},
   journal={Physical Review Letters},
   publisher={American Physical Society (APS)},
   author={Scarani, Valerio and Renner, Renato},
   year={2008},
   month=May }

\PublishersNote{}
\end{document}